\documentclass[letterpaper]{article} 
\usepackage{aaai25}  
\usepackage{times}  
\usepackage{helvet}  
\usepackage{courier}  
\usepackage[hyphens]{url}  
\usepackage{graphicx} 
\usepackage{natbib}  
\usepackage{caption} 
\usepackage{soul}
\usepackage{color}

\usepackage{newfloat}
\usepackage{listings}
\DeclareCaptionStyle{ruled}{labelfont=normalfont,labelsep=colon,strut=off} 
\usepackage{amsmath}
\usepackage{amsfonts}
\usepackage{amsthm}
\usepackage{mathtools}
\usepackage{booktabs}
\usepackage{subcaption}
\usepackage[ruled,vlined]{algorithm2e}

\newcommand{\vx}{\mathbf{x}}
\newcommand{\eps}{\varepsilon}

\newcommand{\ft}{1}

\newcommand{\bx}{X}
\newcommand{\by}{Y}
\newcommand{\bz}{Z}

\newcommand{\df}{q}
\newcommand{\dr}{p}

\newcommand{\ty}{\tilde{\by}}
\newcommand{\tz}{\tilde{\bz}}

\newcommand{\tf}{\tilde{\df}}
\newcommand{\tr}{\tilde{\dr}}

\newcommand{\ts}{\tilde{s}}
\newcommand{\tI}{\tilde{I}}
\def\gL{{\mathcal{L}}}

\newcommand{\betaF}{\beta}
\newcommand{\betaFtx}[2]{\betaF(#1, #2)}

\DeclarePairedDelimiter\abs{\lvert}{\rvert}
\DeclarePairedDelimiter\norm{\lVert}{\rVert}

\makeatletter
\let\oldabs\abs
\def\abs{\@ifstar{\oldabs}{\oldabs*}}

\let\oldnorm\norm
\def\norm{\@ifstar{\oldnorm}{\oldnorm*}}
\makeatother

\newcommand{\expec}[2]{\mathbb{E}_{#1}\left[#2\right]}

\newtheorem{theorem}{Theorem}
\newtheorem{lemma}[]{Lemma}
\newtheorem{proposition}[]{Proposition}

\title{
    Single Exposure Quantitative Phase Imaging with a Conventional Microscope using Diffusion Models
}

\author{
    Gabriel della Maggiora
    \textsuperscript{\rm \equalcontrib,1,2,6},
    Luis Alberto Croquevielle
    \textsuperscript{\rm \equalcontrib,3},
    Harry Horsley\textsuperscript{\rm 4}, \\
    Thomas Heinis\textsuperscript{\rm 3},
    Artur Yakimovich\textsuperscript{\rm 1,2,5 \thanks{Corresponding author.}}
}
\affiliations{
    \textsuperscript{\rm 1}
    Center for Advanced Systems Understanding (CASUS), Görlitz, Germany \\
    \textsuperscript{\rm 2}
    Helmholtz-Zentrum Dresden-Rossendorf e. V. (HZDR), Dresden, Germany \\
    \textsuperscript{\rm 3}
    Department of Computing, Imperial College London, London, United Kingdom \\
    \textsuperscript{\rm 4}
    Bladder Infection and Immunity Group (BIIG), UCL Centre for Kidney and Bladder Health,
    Division of Medicine, University College London,
    Royal Free Hospital Campus, London, United Kingdom \\
    \textsuperscript{\rm 5}
    Institute of Computer Science, University of Wrocław, Wrocław, Poland \\
    \textsuperscript{\rm 6}
    School of Computation, Information and Technology, Technical University of Munich, Germany\\    \{g.della-maggiora valdes,a.yakimovich\}@hzdr.de, aac622@ic.ac.uk, h.horsley@ucl.ac.uk
}

\begin{document}

\maketitle

\begin{abstract}
    Phase imaging is gaining importance due to its applications in fields like biomedical imaging and material characterization. In biomedical applications, it can provide quantitative information missing in label-free microscopy modalities. One of the most prominent methods in phase quantification is the Transport-of-Intensity Equation (TIE). TIE often requires multiple acquisitions at different defocus distances, which is not always feasible in a clinical setting, due to hardware constraints. To address this issue, we propose the use of chromatic aberrations to induce the required through-focus images with a single exposure, effectively generating a through-focus stack. Since the defocus distance induced by the aberrations is small, conventional TIE solvers are insufficient to address the resulting artifacts. We propose Zero-Mean Diffusion, a modified version of diffusion models designed for quantitative image prediction, and train it with synthetic data to ensure robust phase retrieval. Our contributions offer an alternative TIE approach that leverages chromatic aberrations, achieving accurate single-exposure phase measurement with white light and thus improving the efficiency of phase imaging. Additionally, we present a new class of diffusion models that are well-suited for quantitative data and have a sound theoretical basis. To validate our approach, we employ a widespread brightfield microscope equipped with a commercially available color camera. We apply our model to clinical microscopy of patients' urine, obtaining accurate phase measurements.
\end{abstract}

\section{Introduction}

Phase imaging has emerged as a crucial technique in many applications, including biomedical imaging \cite{park_quantitative_2018}, label-free imaging characterization \cite{rivenson_phasestain_2019}, and modality imaging conversion \cite{wang_use_2024, zuo_transport_2020}. It enhances contrast for objects with little or no absorption, such as biological structures, and provides quantitative information about the morphology of microscopic objects like cells, which is typically missing in label-free microscopy \cite{zangle2014live}. For example, this information can be used to detect morphological changes in cells \cite{yakimovich2018label}.

Several techniques have been developed for phase quantification. Among these, the Transport-of-Intensity Equation (TIE) \cite{Teague:83} stands out due to its straightforward approach. TIE does not require phase unwrapping and is easy to integrate into hardware \cite{GUPTA2020126347}. It enables phase retrieval under partially coherent illumination \cite{PhysRevLett.80.2586}, which is common in clinical microscopes. However, a primary challenge in TIE is the presence of low-frequency artifacts \cite{wu_physics-informed_2022, zhu_low-noise_2014}, often requiring multiple through-focus images, which can be time-consuming and technically demanding to get \cite{Waller:10, Zuo2013}. Alternatively, regularization techniques \cite{lustig_sparse_2007, metzler2018prdeep} can be employed, but they are tailored to specific acquisition scenarios and may not capture all image features \cite{zuo_transport_2020}.

In this work, we propose to perform phase imaging by leveraging chromatic aberrations produced in systems with broad-spectrum sources, enabling precise phase estimation with a single exposure. This approach can be implemented with a conventional microscope equipped with a commercially available polychromatic (RGB) camera, making it highly practical for clinical use. However, the use of white light introduces two challenges: a fixed effective defocus distance \cite{waller_phase_2010}, and blurring in the captured diffraction patterns due to source incoherence \cite{zuo_transport_2020, zuo_transport_2015}. To address these issues, we propose a data-driven approach using a deep diffusion model \cite{ho_denoising_2020, saharia_palette_2022} trained on synthetic data \cite{ipads_para}, reducing the reliance on real-world acquisitions \cite{Maggiora2022, Yang2023}. This model provides a robust and practical solution for phase retrieval with polychromatic sources. Our main contributions are:
\begin{itemize}
    \item We propose a novel training paradigm for phase retrieval based on physics-based synthetic data generation \cite{ipads_para}.
    \item We present a flexible and practical method for clinical phase imaging.
    \item We introduce Zero-Mean Diffusion, a variation of diffusion models, along with theoretical insights that explain its improved performance for phase retrieval tasks.
\end{itemize}

\section{Related Work}
\label{sec:related-work}

The induced phase of an object's incident wave $\varphi$ can be determined by analyzing its measured diffraction pattern $I$. For a monochromatic plane wave with wavelength $\lambda$ and  wavenumber $k=\frac{2\pi}{\lambda}$, the Transport-of-Intensity Equation (TIE), establishes a connection between the rate of change of the image in the propagation direction and the lateral phase gradient \cite{zuo_transport_2020}:
\begin{equation}
\label{eq:tie-model}
    -k\frac{\partial I(x,y;z)}{\partial z} = \nabla_{(x,y)}{\cdot}\bigg [I(x,y;z)\nabla_{(x,y)}\varphi(x,y;z)\bigg]
\end{equation}
where $(x, y)$ are lateral spatial coordinates and $z$ is the defocus distance, and $z=0$ is the center of the object. 
Equation (\ref{eq:tie-model}) can be solved for $I$ within the near field of a propagated distance $z$ by the Fresnel diffraction integral. For a wavelength $\lambda$ and denoting $\vx = (x,y)$ this gives
\begin{equation}
\label{eq:fresnel-k}
    I(\vx; z, \lambda)
    =
    \abs{
        A(\vx)e^{i\varphi(\vx)}
        \ast
        \frac{e^{ikz}}{i\lambda z}e^{i\frac{k}{2z}\norm{\vx}_2^2}
    }^2,
\end{equation}
where $A(\vx)$ denotes the amplitude. To recover the phase $\varphi$ from (\ref{eq:tie-model}), two approaches exist. First, assuming the object is predominantly non-absorptive allows for treating it as a phase object, such that $I(\vx;z=0) = I_0$ and (\ref{eq:tie-model}) becomes:
\begin{equation*}
    -k\frac{\partial I(\vx;z)}{\partial z} = I_0\nabla^2_{\vx}{\varphi(\vx;z)}.
\end{equation*}

On the other hand, for an absorptive object, Teague's assumption \cite{Teague:83} yields an alternative formulation. In both cases, we can solve for the phase $\varphi$ by using the Fourier transform in the lateral spatial coordinates and differential operators acting on $I(x, y; z)$ (see details in Appendix \ref{sec:appendix-tie}).

Important numerical challenges render phase recovery an ill-posed problem. Techniques such as Tikhonov regularization \cite{Sixou2013}, Total Variation \cite{cheng_tv}, or Gaussian smoothing \cite{Nakajima1998} can mitigate instabilities, but they complicate the retrieval of low spatial frequency components. Additionally, approximating $\frac{\partial I(x,y;z)}{\partial z}$ through finite differences is susceptible to noise: a small distance yields better detail recovery at the cost of worse Signal-to-Noise Ratio (SNR), while larger distances worsen the paraxial approximation, under which the TIE is valid. \cite{zhu_low-noise_2014, zuo_transport_2020}. These issues can be addressed by increasing the number of acquisitions. While effective, this adds overhead to the method, making it less practical for real-world applications.

Another solution is to employ polychromatic illumination, such as white light. This configuration is useful for phase quantification because the intensity observed at the defocus plane can be interpreted as imaging at multiple defocus distances \cite{zuo_transport_2020}, as we show in Figure \ref{fig:4fsys}. Nevertheless, a naive application of this idea results in a blurred diffraction pattern due to the superposition of the different defocus point spread functions at varying propagation distances. There is also blurring induced by temporal coherence arising from the superposition of diffuse spots at different axial propagation distances, resulting in an image with fewer details. Nonetheless, this configuration can work by using the wavelength-dependent characteristics of Fresnel diffraction and an RGB camera \cite{waller_phase_2010, Zuo2015}. This results in a single-exposure phase quantification method that effectively alleviates the problem of mechanical defocusing in TIE phase retrieval.

Deep learning has recently emerged as a powerful tool for addressing image reconstruction challenges \cite{wang2024multiplephysicsinformedsyntheticdata}, including super-resolution \cite{li_srdiff_2021}, inpainting \cite{Lugmayr2022}, and the reconstruction of under-sampled MRI \cite{Ahishakiye2021} and CT images \cite{Li2022}. In the field of Quantitative Phase Imaging (QPI) \cite{wang_use_2024}, deep learning methods have demonstrated superior phase estimation compared to traditional 2-shot and single-image QPI techniques \cite{zhang_phasegan_2021, wu_physics-informed_2022}. These image reconstruction problems often focus on minimizing the $L_1$ and $L_2$ norms, or other distance metrics. Despite their effectiveness, these methods usually provide blurry solutions and lack a way of modeling the uncertainty in the model \cite{lehtinen2018noise2noise}.

Generative models have gained popularity because they can produce more realistic reconstructions than regression methods. Additionally, they 
model the uncertainty in ill-posed inverse problems, providing a way to sample the probability distribution associated with the problem. Diffusion models are generative likelihood-based models that have been successfully applied to problems like under-sampled MRI image reconstruction \cite{chung2022scorebased}, deblurring \cite{li2023diffusion}, image denoising \cite{chung2023diffusion}, and phase retrieval \cite{dellamaggiora2024conditional, zhang_phasegan_2021}, showing great stability during training and providing state-of-the-art sample quality.

To obtain the phase of an image, we propose using white light as the image source with a polychromatic (RGB) camera-based acquisition. Conveniently, microscopes with such setups are commercially available and widespread in clinical imaging. Since white light is polychromatic, we can quantify phase using the induced defocus stack. Solving the Transport-of-Intensity Equation (TIE) with these measurements is possible in principle because different wavelengths can be translated to different defocus distances under the same wavelength \cite{zuo_transport_2020, waller_phase_2010}. Mathematically, the TIE model given by (\ref{eq:tie-model}) can be reformulated using an auxiliary variable $\xi = \lambda z$, since both $z$ and $\lambda$ are related in equation (\ref{eq:tie-model}). Consequently, the TIE can be expressed in terms of $\xi$ as follows \cite{waller_phase_2010}:
\begin{equation}
\label{eqn:tie-xi}
    -2\pi
    \frac{\partial I(\vx;\xi)}{\partial \xi}
    =
    \nabla_{\vx}{\cdot}
    \big[
        I(\vx;\xi)\nabla_{\vx}\varphi(\vx;\xi)
    \big].
\end{equation}
The main advantage of this formulation is that it allows to perform phase recovery while maintaining a constant defocus plane $z$ \cite{waller_phase_2010}. It also maintains the advantages of the previous TIE formulation, such as only requiring partial temporal coherence. Finally, to solve the equation for $\varphi$, we can use the
color filtering that any digital color camera provides and obtain the required differentially defocused intensity images without the problems introduced by capturing multiple defocus planes.

To circumvent the issues present with using chromatic aberrations as the defocus stack, we propose the use of a diffusion model to solve this problem as an instance of conditioned probabilistic sampling and quantify the phase using polychromatic image acquisition with a white light-based source, achieving high accuracy.

\section{Methods}

Sections \ref{sec:methods-tie} and \ref{sec:methods-dataset} describe our formulation for TIE and the design of a sample dataset for learning phase recovery. Sections \ref{sec:methods-cvdm} and \ref{sec:methods-zmd} describe our use of diffusion models to solve phase recovery as an instance of conditioned sampling.

\subsection{Forward Model}
\label{sec:methods-tie}

\begin{figure}
    \centering
    \includegraphics[width=0.45\textwidth]{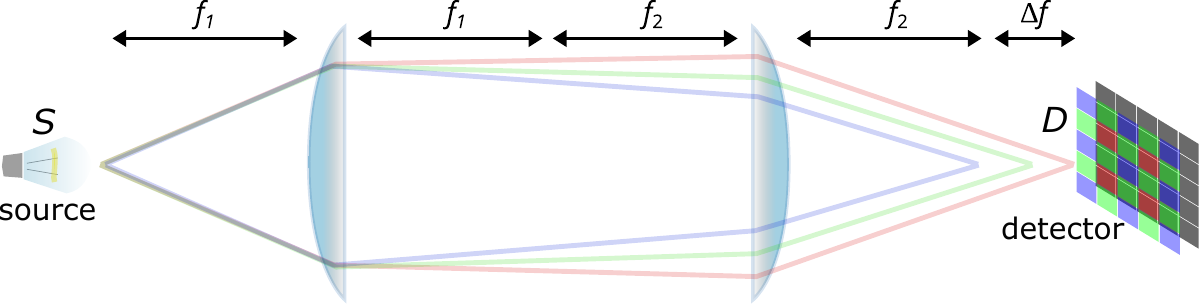}
    \caption{Example of optical system. The defocus distance $\Delta f$ induced by chromatic aberration depends on the focal lengths $f_1$ and $f_2$. The resulting $\Delta f$ for a specific wavelength $\lambda$ corresponds to variable $z$ in (\ref{eq:fresnel-k}).}
    \label{fig:4fsys}
\end{figure}

Consider a thin object $T(\vx) = A(\vx)e^{i\phi(\vx)}$, measured in $z=0$. The equation for the intensity image under partially incoherent illumination \cite{zuo_transport_2020} is 
\begin{align*}
    I(\vx& ; z,\lambda) \\
    &=
    \int S(\mathbf{u})\bigg |
        \int T(\vx')h_{z,\lambda}(\vx-\vx')e^{i2\pi\mathbf{u}\vx'}\mathbf{d}\vx'
    \bigg|^2 \mathbf{d}\mathbf{u} \\
    &=
    \int S(\mathbf{u})I_{\mathbf{u}}(\vx)\mathbf{d}\mathbf{u},
\end{align*}
where $S(\mathbf{u})$ and $I_{\mathbf{u}}$ correspond to the Fourier transforms of the source and the image, respectively, and $h$ is the system's impulse response. This suggests that the observed intensity image is a superposition of all the coherent images obtained from the source. Under the paraxial approximation and assuming that the source is a coherent plane wave, that is, $S(\mathbf{u})=\delta(\mathbf{u})$, the image can be obtained by the Fresnel integral given by equation (\ref{eq:fresnel-k}).

In polychromatic acquisitions, the sensor is usually sensitive to a specific bandwidth. Generally speaking, a polychromatic sensor is built as an array of monochromatic sensors sensitive to three main components: red, green, and blue. Each has a different sensitivity to each wavelength. The sensor's quantum efficiency determines how much energy each frequency transforms into a signal. Given this model, the system can be expressed as a function of the power spectral distribution (PSD) $S_\lambda(\mathbf{u})$ \cite{zuo_transport_2020}. Then, the intensity image is given by
\begin{align}
\label{eqn:fw-model}
    I(\vx& ;z)\notag \\
    &=
    \iint S_\lambda(\mathbf{u})
    \bigg |
        \int T(\vx')h_{z,\lambda}(\vx-\vx')e^{i2\pi\mathbf{u}\vx'}\mathbf{d}\vx'
    \bigg|^2 d\mathbf{u}\mathbf{d}\lambda \notag  \\
    &=
    \iint S_\lambda(\mathbf{u})I_{\mathbf{u},\lambda}(\vx)\mathbf{d}\mathbf{u}d\lambda.
\end{align}
This means that the resulting intensity image is equal to the superposition of the different power densities within the defined bandwidth, i.e., the average of all images produced by the different coherent wavelength components.

\subsection{Dataset Simulation}
\label{sec:methods-dataset}

To simulate polychromatic acquisitions, we use the ImageNet dataset \cite{Deng2009}. Each image undergoes a grayscale conversion and is treated as a pure phase object. Then, we simulate the image under a fixed defocus distance $z$ randomly selected between $0.1\mu$m and $3\mu$m. For the simulation, we use the Fresnel propagator for each wavelength ranging from $400$nm to $700$nm with equal contribution, discretizing with $6$nm increments. For each increment, the source is viewed as a coherent plane wave. Since the quantum efficiency of the sensor influences the signal captured by the detector, we simulate a Gaussian distribution for each channel centered at $\lambda_r = 630$nm, $\lambda_g = 550$nm, and $\lambda_b = 450$nm respectively. We model the quantum efficiency of the sensor using the following function:
\begin{equation*}
    Q_c(\lambda) = e^{-\frac{(\lambda - \lambda_c)^2}{2\sigma_c^2}}
\end{equation*}
where $\sigma_c$ is a parameter representing the sensitivity of the sensor to the corresponding wavelength. To model the PSD for a specific channel $c$ we assume the model $S_\lambda(\mathbf{u})=\delta_\lambda(\mathbf{u})Q_c(\lambda)$. Replacing in (\ref{eqn:fw-model}), we get a single integral with respect to the $\lambda$ variable, which is approximated as:
\begin{equation*}
    I_c(\vx; z)
    = \frac{1}{W}\sum_{i=1}^{W} Q_c(\lambda_i) I(\vx;z,\lambda_i),
\end{equation*}
where $c$ denotes the respective channel. The wavelength range is divided into $6$nm-sized intervals, and the left endpoint $\lambda_i$ and size $1/W$ of each interval are used for a Riemann sum approximation of the integral in (\ref{eqn:fw-model}). Finally, white Gaussian noise is added to the image to simulate the noise in real acquisitions. This procedure allows us to construct a dataset of polychromatic acquisitions along with the correspondent phase of the images, which we use to train a deep learning model for phase reconstruction. Our methodology is based on diffusion models and explained in Sections \ref{sec:methods-cvdm} and \ref{sec:methods-zmd}. See Figure \ref{fig:sim-pipeline} in the appendix for more details.

\subsection{Conditional Denoising Diffusion Models}
\label{sec:methods-cvdm}

Phase recovery is considered ill-posed, meaning that a given acquisition may be consistent with multiple different phases. Hence, we frame the problem as an instance of conditioned sampling, which explicitly accounts for the uncertainty of the recovery. Suppose we have access to a dataset of input-output pairs, denoted as $\mathcal{D} = \{(\bx_i, \by_i)\}_{i=1}^N$. We model these samples as coming from an unknown joint distribution $\df_\text{data}(\bx, \by)$. In our context, $\bx$ corresponds to a sample drawn from polychromatic image acquisition, and $\by$ corresponds to the phase of the image.

Our objective then becomes to learn an approximate way of sampling $\by \sim \df_\text{data}(\cdot | \bx)$, for any given $\bx$. To do this we propose the use of Conditional Variational Diffusion Models (CVDM) \cite{dellamaggiora2024conditional}. In this framework, the basic setup is the existence of a \textit{forward process} $\{\by_t\}_{t\in [0, \ft]}$, where $\by_0 \sim \df_{\text{data}}(\cdot | \bx)$ and the evolution of the process is given by the following stochastic differential equation (SDE) \cite{song_score-based_2021}:
\begin{equation}
\label{eq:forward-sde}
    d\by_t = -\frac{1}{2}\beta(t, \bx)\by_t dt + \sqrt{\beta(t, \bx)}dW_t,
\end{equation}
where $\beta(t, \bx)$ is a \textit{variance schedule} function. The intuitive idea is that noise is gradually added to an initial sample $\by_0 \sim \df_{\text{data}}(\cdot | \bx)$, such that $\by_\ft$ at the end of the process is close to a $\mathcal{N}(0, I)$ variable. For this process, we denote the density at time $t$ by $\df_t(\cdot | \bx)$. The interesting part is learning to (approximately) reverse this forward process. This gives rise to a \textit{reverse process} $\{\bz_t\}_{t\in[0, \ft]}$, whose density at time $t$ we denote by $\dr_t(\cdot | \bx)$, where we aim to have $\dr_t \approx \df_{\ft-t}$.

Specifically, the reverse process starts with $\dr_0(\cdot | \bx) = \df_\ft(\cdot | \bx)$. It is known that $\df_\ft$ approximates a $\mathcal{N}(0, I)$ distribution, so this step is usually implemented as $\bz_0 \sim \mathcal{N}(0, I)$. After sampling $\bz_0$, the reverse dynamics (given by a different SDE) are simulated until we get $\bz_\ft \sim \dr_\ft(\cdot | \bx) \approx \df_\text{data}(\cdot | \bx)$. This can be done in different ways \cite{ho_denoising_2020, song_score-based_2021}, usually via a score or noise-prediction model. In CVDM, a noise-prediction model is learned, and the reverse process is simulated via ancestral sampling \cite{dellamaggiora2024conditional}.

We now explain the way in which CVDM learns and simulates the reverse process. Given equation (\ref{eq:forward-sde}) it can be shown \cite{dellamaggiora2024conditional} that the process $\{\by_t\}$ conditioned on $\by_0$ evolves such that
\begin{equation}
\label{eq:forward-density-t-conditioned-on-0}
    \df_t(\cdot | \by_0 = \by, \bx)
    =
    \mathcal{N}\left(\sqrt{\gamma(t, \bx)} \by, (1-\gamma(t, \bx)) I\right)
\end{equation}
where $\gamma(t, \bx) = e^{-\int_0^t \beta(s, \bx) ds}$. Here, $\gamma$ is also referred to as a \textit{variance schedule} function. It captures the longer-term dynamics of the forward process, while $\beta$ controls the instantaneous rate of change according to (\ref{eq:forward-sde}). Interestingly, for any $s < t$ we have that $\df_t(\cdot | \by_s = y_s, \bx)$ is also normally distributed, and Bayes' theorem implies that the posterior distribution is a Gaussian too \cite{kingma2021variational, dellamaggiora2024conditional}:
\begin{align}
    \df_s(\cdot |& \by_t = y_t, \by_0 = \by, \bx) \nonumber \\
    &=
    \mathcal{N}\Big(
        \alpha(y_t, \by, s, t, \bx),
        \sigma(s, t, \bx)I
    \Big)
    \label{eq:forward-density-s-conditioned-on-t-0}
\end{align}
for $s < t$. Here, the functions $\alpha$ and $\sigma$ have an explicit form in terms of the schedule functions $\beta$ and $\gamma$. In CVDM, this motivates the modeling of the reverse process such that
\begin{align}
    \dr_{\ft-s}(\cdot |& \bz_{\ft-t} = y_t, \bx) \nonumber \\
    &=
    \mathcal{N}\Big(
        \alpha(y_t, \by_\nu, s, t, \bx),
        \sigma(s, t, \bx)I
    \Big)
    \label{eq:backward-density-t-conditioned-on-s}
\end{align}
for $s<t$. The only difference between (\ref{eq:forward-density-s-conditioned-on-t-0}) and (\ref{eq:backward-density-t-conditioned-on-s}) is that $\by$ is unavailable in (\ref{eq:backward-density-t-conditioned-on-s}) and hence replaced by a learned model $\by_\nu(y_t, t, \bx)$. Equation (\ref{eq:backward-density-t-conditioned-on-s}) is the key tool used in CVDM to simulate the reverse process. Concretely, first a partition $(t_0, \ldots, t_N)$ of $[0, \ft]$ is chosen. Then, $\bz_0 \sim \mathcal{N}(0, I)$ is sampled. Finally, equation (\ref{eq:backward-density-t-conditioned-on-s}) is used to iteratively sample $\bz_{t_{i+1}}$ given $\bz_{t_i}$, until $\bz_\ft = \bz_{t_N}$ is obtained.

The main thing left to explain is how the predictor $\by_\nu$ is trained. First, notice that equation (\ref{eq:forward-density-t-conditioned-on-0}) leads to the following parameterization. Conditioned on $\by_0=\by$, the variable $\by_t$ can be written as
\begin{equation}
\label{eq:forward-reparameterization}
    \by_t(\by, \eps)
    =
    \sqrt{\gamma(t, \bx)} \by + \sqrt{1-\gamma(t, \bx)}\eps
\end{equation}
where $\eps\sim\mathcal{N}(0, I)$. A natural approach to learn $\by_\nu$ is to maximize the log-likelihood term $\expec{(\bx, \by) \sim \df_\text{data}}{\dr_\ft(\by | \bx)}$. CVDM and other diffusion models use a variational lower bound to get a more tractable expression. In the end, a diffusion loss term is obtained, which after some simplification corresponds to \cite{dellamaggiora2024conditional}:
\begin{equation*}
    \mathcal{L}_\text{diff}(\bx, \by)
    =
    \frac{1}{2}
    \expec{
        \eps  
        ,t  
    }{
        \norm{\by - \by_\nu(\by_t(\by, \eps), t, \bx)}_2^2
    }.
\end{equation*}
This measures how well we can guess $\by$ given the value of $\by_t$, when the process starts from $\by_0 = \by$. The goal then becomes to minimize $\expec{(\bx, \by) \sim \df_\text{data}}{\mathcal{L}_\text{diff}(\bx, \by)}$. Experimental results \cite{ho_denoising_2020} show that re-parameterizing the model to guess the Gaussian noise $\eps$ improves accuracy. Hence, CVDM learns a noise-prediction model $\eps_\nu$, via a loss term $\mathcal{L}_\text{noise}$ equivalent to $\mathcal{L}_\text{diff}$.

A key advantage of CVDM in comparison to SR3 \cite{saharia2021image} and other conditioning diffusion methods \cite{dhariwal_diffusion_2021} is that it allows for automatically learning the schedule functions. This avoids the need to fine-tune $\beta$ and $\gamma$ by hand, which is costly. However, when learning these functions it is necessary to ensure some basic properties, which are encoded in two loss terms $\mathcal{L}_\beta(\bx)$ and $\mathcal{L}_\text{prior}(\bx, \by)$.
Moreover, a regularization term $\mathcal{L}_\gamma(\bx)$ is needed to avoid pathological schedules. The final loss function used in CVDM is hence given by
\begin{align*}
    \mathcal{L}_\text{CVDM}
    &=
    \mathbb{E}
    _{(\bx,\by) \sim \df_\text{data}}
    \big[ \\
    &\mathcal{L}_\beta(\bx)
        + \mathcal{L}_\text{prior}(\bx, \by)
        + \mathcal{L}_\text{noise}(\by, \bx)
        + a\mathcal{L}_{\gamma}(\bx)
    \big],
\end{align*}
where $a$ is the weight of the regularization term $\mathcal{L}_{\gamma}$. In practice, a Monte Carlo estimator of $\mathcal{L}_\text{CVDM}$ is optimized by using the available $\{(\bx_i, \by_i)\}_{i=1}^N$ samples. For more details, see Appendix \ref{sec:appendix-cvdm}.

\subsection{Zero-Mean Diffusion}
\label{sec:methods-zmd}

A standard practice in diffusion models is to normalize all data to the $[-1, 1]$ interval. This improves numerical stability but becomes an issue for quantitative imaging because the reconstructed data loses direct physical interpretation. In this work, we propose an alternative to normalization called Zero-Mean Diffusion (ZMD), which avoids this problem and has an interesting theoretical basis.

Concretely, we use a simple regression model to learn the expected value of the sample data, in conjunction with a diffusion model which predicts the difference between the regression model and the actual sample. This type of strategy has been used before, with several works showing that modeling residuals instead of the full objective can be simpler and lead to better results \cite{li_srdiff_2021, yang_diffusion_2022}. In addition, theoretical analysis suggests a reason for the high reconstruction quality achieved by ZMD.

As explained in Section \ref{sec:methods-cvdm}, we are interested in sampling from a conditional distribution $\df_\text{data}(\cdot | \bx)$. The approach via diffusion models is to define a stochastic forward process $\{\by_t\}_{t\in [0, \ft]}$, which evolves according to (\ref{eq:forward-sde}) and such that $\df_0(\cdot | \bx) = \df_\text{data}(\cdot | \bx)$. The forward process can be approximately reversed by starting at $\bz_0 \sim \dr_0(\cdot | \bx)$ with $\dr_0 = \mathcal{N}(0, I) \approx \df_\ft$, and then simulating the reverse dynamics until we get $\bz_\ft \sim \dr_\ft(\cdot | \bx)$ with $\dr_\ft \approx \df_0 = \df_\text{data}$, which is the distribution of interest.

We now formally describe ZMD. We define a new stochastic (forward) process $\big \{ \ty_t \big \}$ given by $\ty_t = \by_t - \mu_t$ for $t\in [0, \ft]$ where $\mu_t = \expec{\df_t(\cdot | \bx)}{\by_t}$ is the expected value of $\by_t$. Notice that, similar to our convention with $\by_t$ and $\bz_t$, we omit explicit reference to $\bx$ when talking about $\ty_t$ and $\mu_t$. We do this to simplify notation since the conditioning data $\bx$ remains fixed. An important property of the centered process $\big\{\ty_t\big\}$ is that it also evolves according to (\ref{eq:forward-sde}).
\begin{proposition}
\label{prop:zero-mean-sde}
    The process $\big\{\ty_t\big\}$ given by $\ty_t = \by_t - \mu_t$ evolves according to the following SDE:
    \begin{equation}
    \label{eq:forward-sde-centered}
        d\ty_t = -\frac{1}{2}\beta(t, \bx) \ty_t dt + \sqrt{\beta(t, \bx)} dW_t.
    \end{equation}
\end{proposition}
The proof can be found in Appendix \ref{sec:appendix-proofs}. Proposition \ref{prop:zero-mean-sde} provides an alternative sampling scheme for $\df_0(\cdot | \bx)$ via diffusion models. The idea is to use CVDM as outlined in the previous section but applied now to the $\big\{\ty_t\big\}$ process. We use an analogous notation for this method version so that $\tf_t(\cdot | \bx)$ denotes the density for this process at time $t$. Our main goal is learning to simulate a reverse process $\big\{\tz_t\big\}$, characterized at time $t$ by a density $\tr_t(\cdot | \bx)$ and where the key property is $\tr_t \approx \tf_{\ft-t}$.

Proposition \ref{prop:zero-mean-sde} guarantees that the dynamics of $\big\{\ty_t\big\}$ are exactly the same as those of $\{\by_t\}$, which means CVDM can be implemented in the same way as before, that is, by using the same parameterization and minimizing the loss function $\mathcal{L}_\text{CVDM}$. Nonetheless, there is a caveat. As described in the previous section, in practice, a Monte Carlo estimator of $\mathcal{L}_\text{CVDM}$ is minimized by using the available $\{(\bx_i, \by_i)\}_{i=1}^N$ samples. For ZMD, the Monte Carlo estimator would need access to $\{(\bx_i, \by_i - \mu_0)\}_{i=1}^N$ samples. Hence, ZMD needs the value of $\mu_0$ or at least a good estimate.

Our solution is to define a learnable function $\mu_\theta(\bx)$ to approximate $\mu_0$. To make sure the approximation is accurate, we use the following loss term
\begin{equation*}
    \mathcal{L}_\text{mean}
    =
    \expec{
        (\bx, \by)\sim \df_\text{data}
    }{
        \norm{\by - \mu_\theta(\bx)}_2^2
    },
\end{equation*}
which is now included in the full ZMD loss function:
\begin{equation}
\label{eqn:loss-zmd}
    \mathcal{L}_\text{ZMD}
    =
    \mathcal{L}_\text{CVDM} + \omega\mathcal{L}_\text{mean},
\end{equation}
where $\omega$ controls the relative weights of the terms. For $\mathcal{L}_\text{mean}$, the expected value is calculated with respect to $(\bx, \by)\sim \df_\text{data}$. On the other hand, for $\mathcal{L}_\text{CVDM}$ the expected value is with respect to $(\bx, \ty)\sim \tf_\text{data}$, where $\tf_\text{data}$ represents the distribution of $(\bx, \by - \mu_0)$. This is because CVDM is now learning to reverse the centered process $\big\{\ty_t\big\}$. In practice, we optimize over many iterations. For each one, we simulate $(\bx, \ty)\sim \tf_\text{data}$ by sampling $(\bx, \by)$ from our dataset, and setting $\ty = \by - \mu_\theta(\bx)$. As $\mu_\theta$ improves, this procedure becomes increasingly better at approximating the sampling we need.

Once all models have been trained, we can use ZMD to approximately sample from $\df_0(\cdot | \bx) = \df_\text{data}(\cdot | \bx)$ by following an almost identical procedure as that of standard CVDM:
\begin{enumerate}
    \item Equation (\ref{eq:forward-sde-centered}) implies $\tf_\ft \approx \mathcal{N}(0, I)$. Hence, a good approximation for $\tr_0$ is $\mathcal{N}(0, I)$.
    \item After choosing a time discretization $(t_0, \ldots, t_N)$ of the $[0, \ft]$ interval, equation (\ref{eq:backward-density-t-conditioned-on-s}) can be used used to iteratively sample $\tz_{t_{i+1}}$ given $\tz_{t_i}$, until getting $\tz_\ft = \tz_{t_N}$. If the noise-prediction model $\tilde{\eps}_\nu$ is accurate, then it holds that $\tz_\ft \sim \tr_\ft$ with $\tr_\ft \approx \tf_0$.
\end{enumerate}

The previous two steps produce a sample $\tz_\ft \mathrel{\dot{\sim}} \tf_0$. To get an (approximate) sample from the distribution of interest $\df_0(\cdot | \bx)$, we compute $\tz_\ft + \mu_\theta(\bx)$. This can be understood intuitively. Since $\mu_0$ (or its estimate $\mu_\theta$) is subtracted at the start of the forward process, it must be added back at the end of the reverse process. For the mathematical formalization of this idea, please see Appendix \ref{sec:appendix-proofs}.

As shown in Section \ref{sec:results}, ZMD works very well as an alternative to the standard practice of normalizing data to the $[-1, 1]$ interval. Hence, it provides a robust way of using diffusion models for quantitative data. Besides these strong empirical results, there is a deeper theoretical consideration underlying the performance of ZMD. To explain this, we will use $\tr_\ft^{\mu_0}$ to denote the distribution of the random variable $\tz_\ft + \mu_0$ where $\tz_\ft \sim \tr_\ft$.

Standard CVDM (Section \ref{sec:methods-cvdm}) can be used to sample from a distribution $\dr_\ft$, while ZMD can be used to sample from a distribution $\tr_\ft^{\mu_0}$. Both $\dr_\ft$ and $\tr_\ft^{\mu_0}$ should approximate $\df_0(\cdot | \bx) = \df_\text{data}(\cdot | \bx)$, so it is natural to ask which method leads to a better approximation. To answer this question, we could compare $D_\text{KL}(\df_0 || \dr_\ft)$ with $D_\text{KL}(\df_0 || \tr_\ft^{\mu_0})$. While these divergences cannot be computed analytically, there are known upper bounds for them, which we denote $M(\dr_\ft)$ and $M(\tr_\ft^{\mu_0})$, respectively. Our main theoretical result is that the upper bound for ZMD is always better.

\begin{theorem}
\label{thm:KL-upper-bound-inequality}
    For the stochastic processes as defined in Sections \ref{sec:methods-cvdm} and \ref{sec:methods-zmd}, and for any distribution $\df_0$, it holds that $M(\tr_\ft^{\mu_0}) \leq M(\dr_\ft)$.
\end{theorem}

The proof can be found in Appendix \ref{sec:appendix-proofs}. We conjecture that this theoretical advantage of using a centered forward process may be one of the reasons for the good performance we get when using ZMD. Notice that this only holds if we have access to an accurate model $\mu_\theta$ of $\mu_0$.

\section{Results}
\label{sec:results}

To evaluate our model, we conduct three distinct experiments. First, we generate a set of synthetic samples analogous to those in the training set. Second, we acquire four through-focus stacks (which allow us to compute a ground truth) and subsequently test our method on these images. Finally, we use an uncurated clinical dataset comprised of real-world clinical images. While ground truth data is unavailable for this dataset, we compare our results with reference values reported in the literature. For the first two experiments, we include 2-shot (using two images to measure the phase) acquisition methods as additional reference values. In all reported results, ``2s'' stands for the 2-shot modality, ``p'' for the polychromatic (single exposure) modality, and ``MP'' refers to the metrics obtained by the regression model $\mu_\theta$, which serves as a baseline.

\subsection{Synthetic Images}

To validate the model's performance, we use the HCOCO dataset \cite{hcoco} to simulate synthetic examples. Table \ref{table:qpi-syn} shows the results, and Figure \ref{fig:qpi-synthetic} shows a qualitative comparison of the methods.

\begin{figure}
  \centering
  \includegraphics[width=0.75\columnwidth]{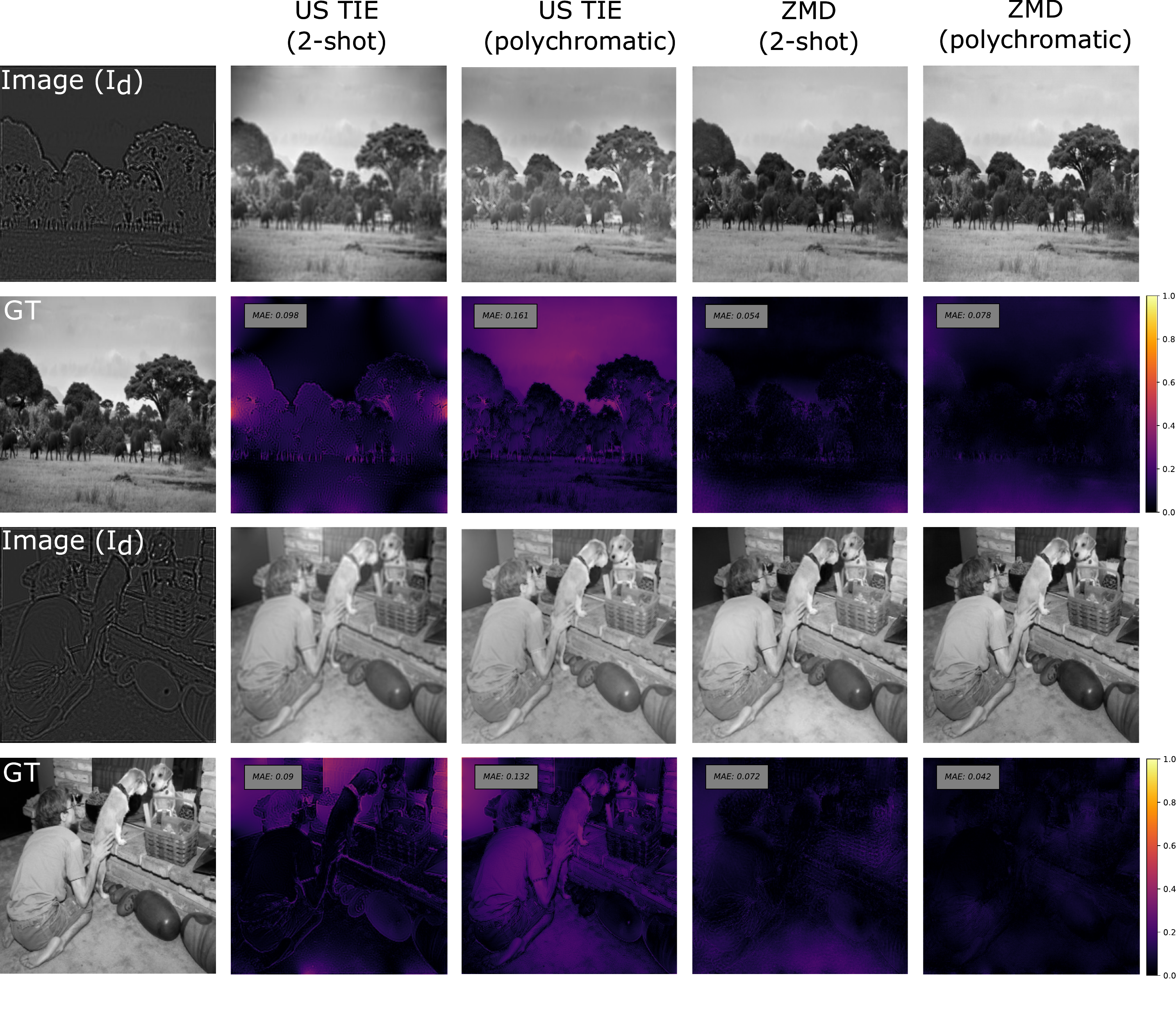}
  \caption{QPI methods assessed using synthetic images from HCOCO.}
  \label{fig:qpi-synthetic}
\end{figure}

\begin{figure}
  \centering
  \includegraphics[width=0.75\columnwidth]{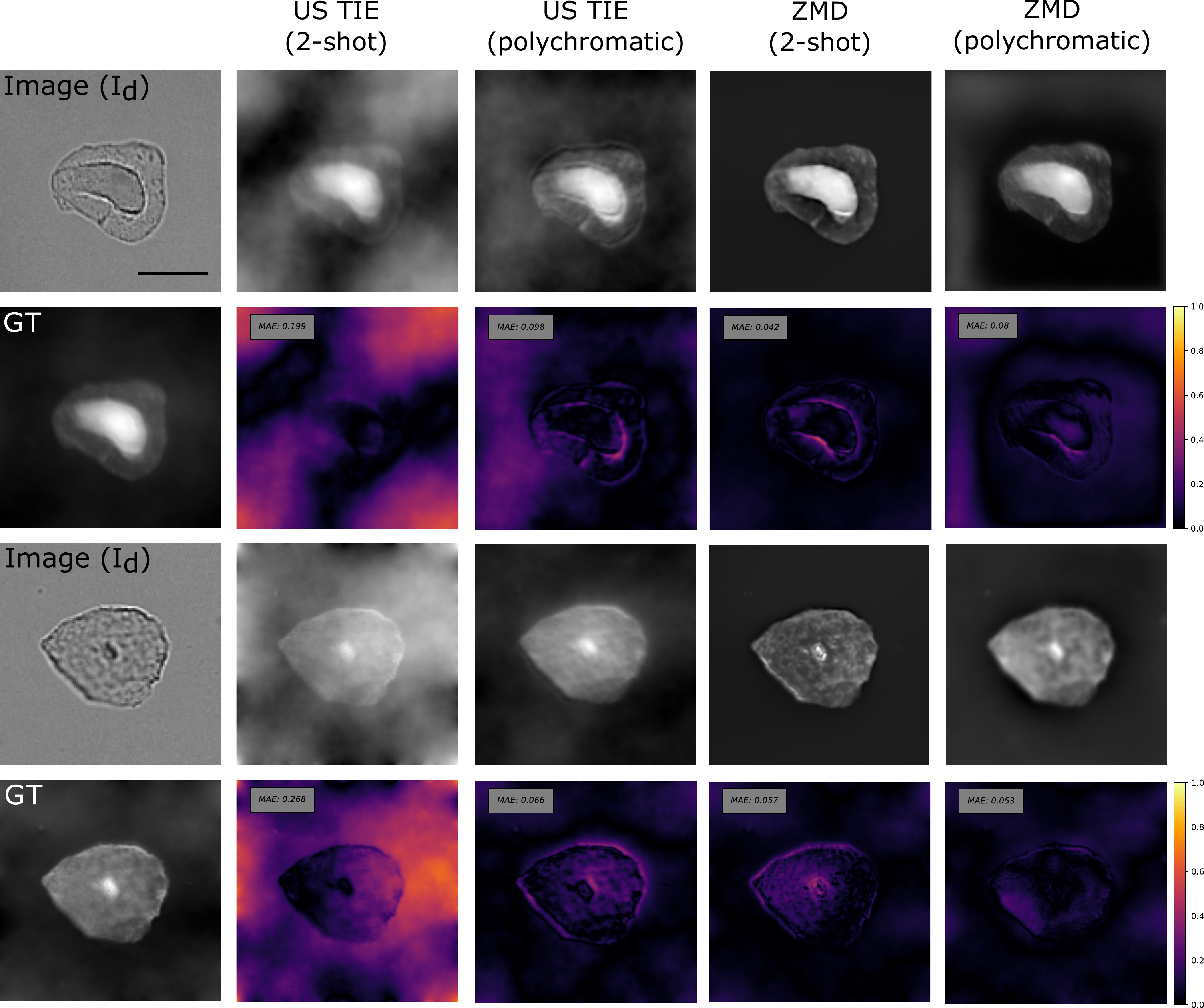}
  \caption{QPI methods assessed using clinical microscopy images depicting two overlapping (two upper rows) and an individual (two lower rows) epithelial cells. Scale bar (upper leftmost image) is 50 $\mu$m.}
  \label{fig:qpi-clinical}
\end{figure}

\begin{table}
    \centering
    \caption{Performance metrics for Synthetic HCOCO (average value and standard deviation over dataset).}
    \label{table:qpi-syn}
    \tiny
        \begin{tabular}{llllll}
            \toprule
            Metric / Model & US TIE (2s) & US TIE (p) &MP (p) & ZMD (2s) & ZMD (p) \\
            \midrule
            MS-SSIM ($\uparrow$) & 0.91 $\pm$ 0.05 & 0.81 $\pm$ 0.09 & 0.95 $\pm$ 0.02&0.94 $\pm$ 0.03 & \textbf{0.97 $\pm$ 0.02} \\
            MAE ($\downarrow$) & 0.11 $\pm$ 0.04 & 0.14 $\pm$ 0.05 & \textbf{0.06 $\pm$ 0.02}&0.08 $\pm$ 0.02 & \textbf{0.06 $\pm$ 0.02} \\
            \bottomrule
        \end{tabular}
\end{table}

\begin{table}
    \centering
    \caption{MS-SSIM on Through-Focus Brightfield Images.}
    \label{table:ssim-through-focus}
    \tiny
    \begin{tabular}{llllll}
        \toprule
        Sample & US TIE (2s) & US TIE (p) &MP (p) & ZMD (2s) & ZMD (p) \\
        \midrule
        1 & 0.61 & 0.65 & 0.78&\textbf{0.91} & 0.82 \\
        2 & 0.63 & 0.78 & 0.88&0.86 & \textbf{0.89} \\
        3  & 0.74 & 0.68 & 0.83&0.81 & \textbf{0.88} \\
        4 & 0.78 & 0.74 & 0.84&0.75 & \textbf{0.86} \\
        \bottomrule
    \end{tabular}
\end{table}

\begin{table}
    \centering
    \caption{MAE on Through-Focus Brightfield Images.}
    \label{table:mae-through-focus}
    \tiny
    \begin{tabular}{llllll}
        \toprule
        Sample & US TIE (2s) & US TIE (p)& MP (p) & ZMD (2s) & ZMD (p) \\
        \midrule
        1 & 0.20 & 0.10 & 0.09 &\textbf{0.04} & 0.08 \\
        2 & 0.27 & 0.07 & 0.06 &0.06 & \textbf{0.05} \\
        3 & 0.18 & 0.13 &\textbf{0.10}  &0.11 & \textbf{0.10} \\
        4 & 0.18 & 0.10&\textbf{0.09}  & 0.10 & \textbf{0.09} \\
        \bottomrule
    \end{tabular}
\end{table}

\begin{figure*}
    \centering
    \includegraphics[width=0.95\textwidth]{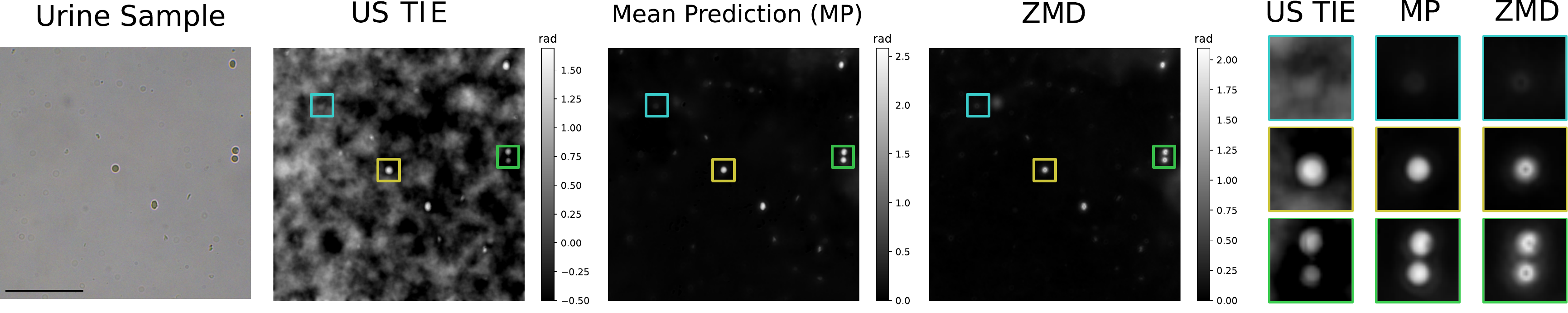}
    \caption{Urine sample with red blood cells. From left to right: US TIE phase reconstruction, Mean Prediction, and Zero-Mean Diffusion prediction. Regions of the image are enlarged to better show reconstruction details. Phase reconstructions are quantitative, with values between 0 and 2.5 radians. Scale bar is 200 $\mu$m.}
    \label{fig:uti-rbc}
\end{figure*}

\subsection{Through-focus Images}
We assessed our approach using four stacks of clinical urine microscopy samples, each acquired as a through-focus series. The ground truth phase was obtained by solving the TIE using US TIE as solver \cite{zhang_universal_2020}. To obtain an accurate ground-truth we estimated $\partial I(x,y; z) / \partial z$ using a 20th-degree polynomial fitted on each pixel of the volume \cite{Waller:10}. The optimization was done considering the images at $d = \pm 2k \mu$m, taking $k$ from 1 to 20. We compared our polychromatic method, utilizing a defocus distance of $z = 2\mu$m, against the 2-shot ZMD and US TIE methods, which were evaluated at $z = \pm 2\mu$m. Diffusion models were trained only on synthetic data. Figures \ref{fig:qpi-clinical} and \ref{subfig:qpi-real-2} (Supplementary Material) display the reconstructions for all stacks, while Tables \ref{table:ssim-through-focus} and \ref{table:mae-through-focus} present the MS-SSIM and MAE scores for each sample and method, corresponding to the rows in the figures. The first two samples are in Figure \ref{fig:qpi-clinical}, and the next two are in Figure \ref{subfig:qpi-real-2}.

\subsection{Clinical microscopy dataset}

To evaluate the applicability of our proposed methodology to a real-world clinical context, we conduct experiments using the Urinary Tract Infection (UTI) dataset as outlined in \cite{liou_clinical_2024}. The images in this dataset were acquired for UTI patient condition evaluation in a point-of-care clinical setting. Image acquisition was performed in a through-focus manner with 200 defocused images separated by 1$\mu$m each using an Olympus BX41F microscope frame, U-5RE quintuple nosepiece, U-LS30 LED illuminator, U-AC Abbe condenser, and two types of objectives: an Olympus Plan 20x/0.40 Infinity/0.17 (used in Figure \ref{fig:qpi-rbc}) and a UPlanFL N 20x/0.5 UIS2 Infinity/0.17/OFN26.5 (used in Figures \ref{fig:qpi-clinical},  \ref{fig:uti-rbc},  \ref{fig:uti-1}, \ref{fig:rbcs} and \ref{subfig:qpi-real-2}). Polychromatic images were taken with a commercially available digital scientific 16-bit color camera (Infinity 3S-1UR, Teledyne Lumenera) connected to the frame of the microscope using a 0.5x C-mount adapter. We calculate the phase on images with both a small and a high number of objects. Figure \ref{fig:uti-1} shows the phase quantification performed by our model. Figure \ref{fig:uti-rbc} shows the reconstruction for red blood cells using three different polychromatic methods: US TIE, Mean Prediction (that is, the regression model $\mu_\theta$ of ZMD without the residual prediction), and ZMD. Crucially, our quantitative results are consistent with the values reported in the literature, that is, between 1.8 and 2.3 rad \cite{Park2019, Nguyen2017}.

\begin{figure}
    \centering
    \includegraphics[width=0.85\columnwidth]{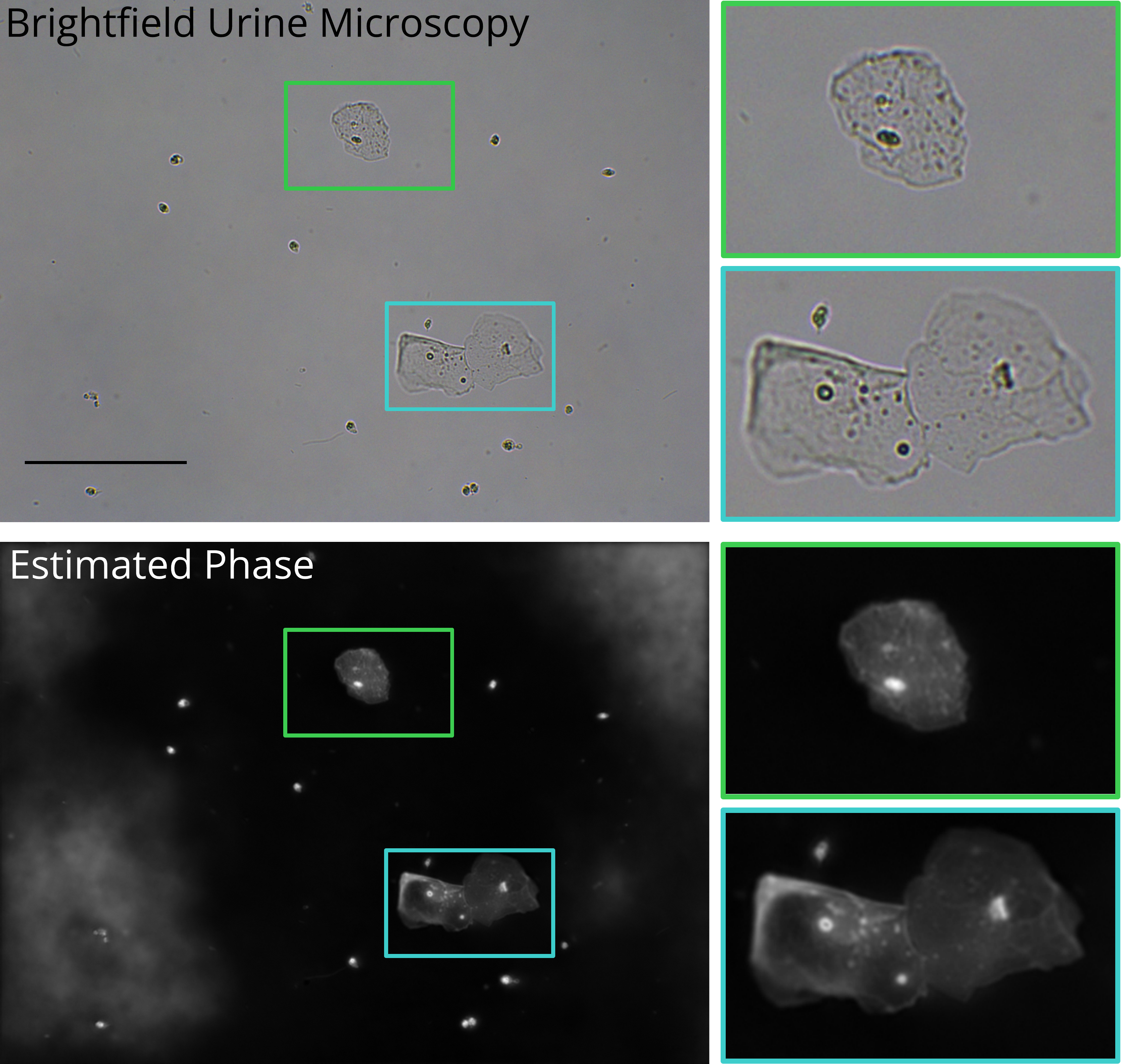}
    \caption{Phase quantification performed on clinical images. Regions of the images are enlarged to better show the objects in the image (epithelial cells). Scale bar (upper leftmost image) is 200 $\mu$m.}
    \label{fig:uti-1}
\end{figure}

\section{Discussion}


Our method was tested in 2-shot and polychromatic modalities, with ZMD showing significantly better performance than US TIE in the latter (see Tables \ref{table:qpi-syn}, \ref{table:mae-through-focus}, \ref{fig:table-ssim-rbc}, \ref{fig:table-mae-rbc}). This highlights ZMD's clinical potential, capturing fine morphometric details of both flake-shaped epithelial cells and disk-shaped red blood cells in urine, including their biconcave morphology and overlapping cytoplasm. Results were validated using two microscopes with different objectives: a UPlanFL N 20x/0.5 UIS2 Infinity/0.17/OFN26.5 objective (Figures \ref{fig:qpi-clinical}, \ref{fig:uti-rbc}, \ref{fig:uti-1}, \ref{fig:rbcs}, and \ref{subfig:qpi-real-2}) and an Olympus Plan 20x/0.40 Infinity/0.17 objective (Figure \ref{fig:qpi-rbc}).


The improvement over the traditional 2-shot modality stems from the model's use of three RGB measurements, enabling better noise handling and more reliable estimation of the through-focus gradient and enhancing image quality.


In our experiments, ZMD recovered more structure than US TIE and Mean Prediction, uniquely revealing the ring cell morphology of red blood cells. As shown in Figures \ref{fig:qpi-clinical} and \ref{fig:uti-rbc}, US TIE produced cloud artifacts. Polychromatic acquisitions inherently exhibit spatial and temporal incoherence, causing some blurring of diffraction patterns and the quantified phase (e.g., Figure \ref{fig:qpi-clinical}). Despite this trade-off, ZMD's improved noise handling and robust performance highlight its potential for medical imaging applications.

\section{Conclusions}

We present a novel method for reliably quantifying phase using chromatic aberrations rather than a purposely acquired through-focus stack. To solve this problem, we apply a denoising diffusion model designed for quantitative tasks. We show that this technique achieves high accuracy using a standard microscope frame with an RGB (color) camera, and our method's applicability is validated on clinical samples. As all necessary hardware is commercially available and commonly present in clinical microscopy setups, our approach holds significant promise for real-world clinical adoption.


Future work should focus on leveraging Physics-Informed Neural Networks (PINNs) to solve the Helmholtz equation, enabling more accurate wave modeling without the paraxial approximation. This would address current limitations and improve simulation realism. Additionally, creating a standardized QPI dataset is crucial for accurate evaluation and benchmarking of different methods.

\newpage

\section*{Acknowledgments}
This work was partially funded by the Center for Advanced Systems Understanding (CASUS), which is financed by Germany’s Federal Ministry of Education and Research (BMBF) and by the Saxon Ministry for Science, Culture, and Tourism (SMWK) with tax funds on the basis of the budget approved by the Saxon State Parliament. The authors acknowledge the support of SMWK: ScaDS.AI and Chile's National Agency for Research and Development (ANID) through their Scholarship Program (DOCTORADO BECAS CHILE/2023 - 72230222).


\bibliography{aaai25}

\appendix

\newpage
\section{Solving the Transport-of-Intensity Equation}
\label{sec:appendix-tie}

As explained in Section \ref{sec:related-work}, the Transport-of-Intensity Equation (TIE) establishes a connection between the rate of change of the image in the propagation direction and the lateral phase gradient, a connection summarized in equation (\ref{eq:tie-model}). This equation can be solved for the phase $\varphi$ by assuming certain properties of the intensity pattern $I$. The simplest case is to consider a pure phase (non-absorptive) object under constant illumination $I_0$. Then, the in-focus intensity image is given by $I(x,y;z=0)=I_0$. In this case, it is possible to solve the equation by employing a pair of Fourier transforms:
\begin{equation*}
    \varphi=\mathcal{F}^{{-}1}\bigg\{\frac{1}{4\pi^2(k_x^2+k_y^2)}\mathcal{F}\bigg\{\frac{-k}{I_0}\frac{\partial I(x,y;z)}{\partial z}\bigg\}\bigg\},
\end{equation*}
where $k_x$ and $k_y$ denote the spatial coordinates in the Fourier space. For a general absorptive object, $I(x,y;z)$ is not a constant and it is necessary to use Teague's assumption to obtain a solution:
\begin{equation}
    \label{eqn:teague-sol}
    \varphi=-k \nabla^{-2}\left\{\nabla \cdot \left\{\frac{1}{I(x,y;z)}\nabla \nabla^{-2}\left [\frac{\partial I(x,y;z)}{\partial z}\right] \right\} \right\},
\end{equation}
where the $\nabla^{-2}$ operator is defined as:
\begin{equation*}
    \nabla^{-2}g(x,y) = \mathcal{F}^{{-}1}\bigg\{\frac{1}{-4\pi^2(k_x^2+k_y^2)}\mathcal{F}\bigg\{g(x,y)\bigg\}\bigg\}.
\end{equation*}

Equation (\ref{eqn:teague-sol}) is only valid when the vector field $I\nabla \varphi = - k \nabla \nabla^{-2} \big[ \frac{\partial I(x,y;z)}{\partial z} \big]$ is conservative. However, this assumption does not always hold. This is especially true in regions of significant absorption changes, such as those caused by dust, rendering the approximation invalid \cite{zhang_universal_2020}. Furthermore, dark regions within $I(x,y;z)$ can introduce instability to the fractional term. Additionally, in both cases, solving TIE implies estimating the $\frac{\partial I(x,y;z)}{\partial z}$ term. This is usually done using finite differences \cite{wu_physics-informed_2022}, or via an interpolated polynomial whose derivative is then evaluated around $z=0$ \cite{Waller:10}. However, since the acquired images have noise, using the outlined solutions produces low-frequency noise amplification around the singularity, rendering phase recovery an ill-posed problem.

\section{More details on Conditional Variational Diffusion Models}
\label{sec:appendix-cvdm}

In this section, we fill in some of the details of the Conditional Variational Diffusion Models (CVDM) framework \cite{dellamaggiora2024conditional}. As mentioned in Section \ref{sec:methods-cvdm}, a key distinction of CVDM in comparison to other conditioning diffusion methods \cite{saharia2021image, dhariwal_diffusion_2021} is that it allows for automatically learning the schedule functions during the training process. This avoids the need to fine-tune $\beta$ and $\gamma$ by hand, a time-consuming task that leads to possibly suboptimal results. On the other hand, learning these functions during training means that some basic properties must be ensured. First, $\beta$ and $\gamma$ are made dependent on the conditioning variable $\bx$ so that the schedule may be tailored to each instance. Second, any \textit{variance-preserving} schedule needs to satisfy the following condition \cite{dellamaggiora2024conditional}:
\begin{equation}
\label{eq:schedule-differential-equation}
    \frac{\partial \gamma(t,\bx)}{\partial t}
    =
    -\beta(t,\bx)\gamma(t,\bx).
\end{equation}
Additionally, from equation (\ref{eq:forward-reparameterization}), we get the two natural conditions $\gamma(0, \bx) \approx 1$ and $\gamma(\ft, \bx) \approx 0$, which ensure that the forward process starts at $\df_\text{data}(\cdot | \bx)$ and ends in $\mathcal{N}(0, I)$, respectively. These conditions are encoded in the following loss function:
\begin{align*}
    \mathcal{L}_\beta(&\bx)= \\
    &\expec{
        t \sim \text{Unif}([0,\ft])
    }{
        \norm{
            \frac{\partial \gamma(t,\bx)}{\partial t}
            + \beta(t, \bx)\gamma(t,\bx)
        }_2^2
    } \\
    & + \norm{\gamma(0, \bx) - 1}_2^2 + \norm{\gamma(\ft, \bx) - 0}_2^2.
\end{align*}
An additional loss term $\mathcal{L}_\text{prior}(\bx, \by)$ is used to ensure $\df_\ft$ is as similar as possible to a $\mathcal{N}(0, I)$ variable. Finally, a regularization term for $\gamma$ is needed to avoid pathological schedules with abrupt changes \cite{dellamaggiora2024conditional}:
\begin{equation*}
    \mathcal{L}_{\gamma}(\bx)
    =
    \expec{
        t \sim \text{Unif}([0,\ft])
    }{
        \norm{\frac{\partial^2\gamma(t, \bx)}{\partial t^2}}_2^2
    }
\end{equation*}
Grouping all these terms, the final loss function used in CVDM is given by
\begin{align*}
    \mathcal{L}_\text{CVDM}
    =
    \mathbb{E}_{
        (\bx,\by) \sim \df_\text{data}
    }
    \Big[
        &\mathcal{L}_\beta(\bx)
        + \mathcal{L}_\text{prior}(\bx, \by) \\
        &+ \mathcal{L}_\text{noise}(\by, \bx)
        + a\mathcal{L}_{\gamma}(\bx)
    \Big],
\end{align*}
where $a$ is the weight of the regularization term $\mathcal{L}_{\gamma}$. As explained in Section \ref{sec:methods-cvdm}, a Monte Carlo estimator of $\mathcal{L}_\text{CVDM}$ is optimized by using the available $\{(\bx_i, \by_i)\}_{i=1}^N$ dataset.

\section{Proofs for Section \ref{sec:methods-zmd}}
\label{sec:appendix-proofs}

In this section, we provide all necessary proofs supporting the theoretical discussion in Section \ref{sec:methods-zmd}. Throughout, our main object of study is the stochastic process $\big\{\ty_t\big\}$ given by $\ty_t = \by_t - \mu_t$. Recall that $\mu_t$ is defined as $\mu_t = \expec{\df_t(\cdot | \bx)}{\by_t}$, that is, as the expected value of $\by_t$ when the forward process is conditioned on $\bx$. We start by proving the following auxiliary result, which gives an explicit expression for $\mu_t$ in terms of $\mu_0$.

\begin{lemma}
\label{lemma:mu-time-t}
    For $t\in [0, \ft]$ it holds that $\mu_t = \mu_0\sqrt{\gamma(t, \bx)}$.
\end{lemma}
\begin{proof}
    Starting from the definition of $\mu_t$, we get
    \begin{align*}
        \mu_t
        &=
        \int y \> \df_t(y | \bx) \> dy \\
        &=
        \int y \int \df_t(y | \by_0 = y_0, \bx) \> \df_0(y_0 | \bx) \> dy_0 \> dy \\
        &=
        \int \df_0(y_0 | \bx) \int y \> \df_t(y | \by_0 = y_0, \bx) \> dy \> dy_0
    \intertext{By equation (\ref{eq:forward-density-t-conditioned-on-0}) we know $\df_t(y | \by_0 = y_0, \bx)$ is Gaussian with expected value $\sqrt{\gamma(t, \bx)} \> y_0$:}
        &=
        \int \df_0(y_0 | \bx) \left(\sqrt{\gamma(t, \bx)} \> y_0 \right) dy_0 \\
        &=
        \sqrt{\gamma(t, \bx)} \int y_0 \> \df_0(y_0 | \bx) \> dy_0 \\
        &=
        \sqrt{\gamma(t, \bx)} \> \mu_0.
    \end{align*}
\end{proof}

We now provide the proof for Proposition \ref{prop:zero-mean-sde}. This result shows that the stochastic process $\big\{\ty_t\big\}$ evolves according to the same stochastic differential equation (SDE) as $\{\by_t\}$.

\begin{proof}[Proof of Proposition \ref{prop:zero-mean-sde}]
    Define the function $F(y, t) = y - \mu_t$. By Ito's Rule \cite{karatzas1991brownian} we know that
    \begin{align*}
        dF(\by_t, t)
        =
        \Bigg[
            &-\frac{1}{2}\beta(t, \bx) \by_t \frac{\partial F}{\partial y}(\by_t, t)
            + \frac{\partial F}{\partial t}(\by_t, t) \\
            &+ \frac{1}{2}\beta(t, \bx) \frac{\partial^2 F}{\partial y^2}(\by_t, t)
        \Bigg] dt \\
        &+ \sqrt{\beta(t, \bx)} \frac{\partial F}{\partial y}(\by_t, t) dW_t.
    \end{align*}
    It is easy to see that $\frac{\partial F}{\partial y} = 1$, $\frac{\partial^2 F}{\partial y^2} = 0$ and $\frac{\partial F}{\partial t} = -\frac{\partial \mu_t}{\partial t}$. Now, using Lemma \ref{lemma:mu-time-t} we get
    \begin{align*}
        -\frac{\partial \mu_t}{\partial t}
        &= -\frac{\partial \left(\mu_0\sqrt{\gamma(t, \bx)}\right)}{\partial t} \\
        &= -\mu_0 \frac{\partial \sqrt{\gamma(t, \bx)}}{\partial t} \\
        &= -\mu_0 \frac{1}{2\sqrt{\gamma(t, \bx)}}
            \frac{\partial \gamma(t, \bx)}{\partial t}
    \intertext{
        Using equation (\ref{eq:schedule-differential-equation}) to replace
        $\frac{\partial \gamma(t, \bx)}{\partial t}$:
    }
        &= \mu_0 \frac{1}{2\sqrt{\gamma(t, \bx)}} \gamma(t, \bx) \beta(t, \bx) \\
        &= \frac{1}{2} \beta(t, \bx) \mu_0 \sqrt{\gamma(t, \bx)} \\
        &= \frac{1}{2} \beta(t, \bx) \mu_t.
    \end{align*}
Replacing all partial derivatives in the expression for $dF(\by_t, t)$ we get:
    \begin{align*}
        dF(&\by_t, t) \\
        &=
        \left[
            -\frac{1}{2} \beta(t, \bx) \by_t
            +\frac{1}{2} \beta(t, \bx) \mu_t
        \right] dt
        + \sqrt{\beta(t, \bx)} dW_t \\
        &=
        -\frac{1}{2} \beta(t, \bx) (\by_t - \mu_t) dt
        + \sqrt{\beta(t, \bx)} dW_t \\
        &=
        -\frac{1}{2} \beta(t, \bx) F(\by_t, t) dt
        + \sqrt{\beta(t, \bx)} dW_t.
    \end{align*}
    Substituting $\ty_t$ for $F(\by_t, t)$ finishes the proof.
\end{proof}

We now work towards proving our main theoretical result, given by Theorem \ref{thm:KL-upper-bound-inequality}. First, we prove a straightforward relationship between the density functions $\{\df_t\}_{t\in [0, \ft]}$ \big(characterizing the standard forward process $\{ \by_t \}_{t\in [0, \ft]}$\big) and the density functions $\{\tf_t\}_{t\in [0, \ft]}$ \big(characterizing the centered forward process $\big \{ \ty_t \big \}_{t\in [0, \ft]}$\big).

\begin{lemma}
\label{lemma:density-tilde}
    For all $t\in [0, \ft]$ it holds that
    \begin{equation*}
        \tf_t(y | \bx) = \df_t(y + \mu_t | \bx).    
    \end{equation*}
\end{lemma}
\begin{proof}
    For a given $\bx$, recall that $\by_t$ is the random variable resulting from sampling $\by_0 \sim \df_\text{data}(\cdot | \bx)$ and then evolving the stochastic process according to the SDE given by (\ref{eq:forward-sde}). Its density is denoted as $\df_t(\cdot | \bx)$. Since $\ty_t$ is defined as $\by_t - \mu_t$ this lemma is a straightforward application of the change-of-variables formula in probability. That is, for any bijective and differentiable transformation $\ty_t = G(\by_t)$ the density $\tf_t(\cdot | \bx)$ of $\ty_t$ is given by:
    \begin{equation*}
        \tf_t(y | \bx)
        =
        \df_t\big( G^{-1}(y) | \bx \big)
        \abs{
            \text{det}\left(
                \mathbf{J}G^{-1}(y)
            \right)
        },
    \end{equation*}
    where $G^{-1}$ is the inverse function and $\mathbf{J}$ denotes the Jacobian operator. In our case, $G^{-1}(y)=y+\mu_t$ and the determinant term is $1$ because $\mu_t$ is fixed (not dependent on $y$). This gives the result.
\end{proof}

We need to prove three more results in a similar vein to Lemma \ref{lemma:density-tilde}. The first one involves the density function when conditioning on $\ty_0$.

\begin{lemma}
\label{lemma:density-tilde-conditioned}
    For all $t\in [0, \ft]$ it holds that
    \begin{equation*}
        \tf_t\big(y | \ty_0=y_0, \bx\big)
        = \df_t(y + \mu_t | \by_0=y_0 + \mu_0, \bx).
    \end{equation*}
\end{lemma}
\begin{proof}
    By Proposition \ref{prop:zero-mean-sde} we know that $\big \{ \ty_t \big \}$ evolves according to the SDE given by (\ref{eq:forward-sde}). Hence, equation (\ref{eq:forward-density-t-conditioned-on-0}) implies that $\tf_t\big(\cdot | \ty_0=y_0, \bx\big)$ has a $\mathcal{N}\big(\sqrt{\gamma(t, \bx)} y_0, (1-\gamma(t, \bx)) I\big)$ distribution. This leads to the following calculation, where we use $\gamma$ in place of $\gamma(t, \bx)$ to simplify notation:
    \begin{align*}
        \tf_t\big(&y | \ty_0=y_0, \bx\big) \\
        &=
        \mathcal{N}\big(y; \sqrt{\gamma} y_0, (1-\gamma) I\big) \\
        &=
        \frac{1}{(2\pi)^{\frac{d}{2}}(1-\gamma)^{\frac{d}{2}}}
        e^{
            -\frac{1}{2(1 - \gamma)}
            (y - \sqrt{\gamma}y_0)^T
            (y - \sqrt{\gamma}y_0)
        } \\
        &=
        \frac{1}{(2\pi)^{\frac{d}{2}}(1-\gamma)^{\frac{d}{2}}}
        e^{
            -\frac{1}{2(1 - \gamma)}
            \norm{y - \sqrt{\gamma}y_0}_2^2
        } \\
        &=
        \frac{1}{(2\pi)^{\frac{d}{2}}(1-\gamma)^{\frac{d}{2}}}
        e^{
            -\frac{1}{2(1 - \gamma)}
            \norm{y + \mu_t - \mu_t - \sqrt{\gamma}y_0}_2^2
        }
    \intertext{By Lemma \ref{lemma:mu-time-t} we know that $\mu_t = \sqrt{\gamma(t, \bx)}\mu_0$:}
        &=
        \frac{1}{(2\pi)^{\frac{d}{2}}(1-\gamma)^{\frac{d}{2}}}
        e^{
            -\frac{1}{2(1 - \gamma)}
            \norm{y + \mu_t - \sqrt{\gamma}\mu_0 - \sqrt{\gamma}y_0}_2^2
        } \\
        &=
        \frac{1}{(2\pi)^{\frac{d}{2}}(1-\gamma)^{\frac{d}{2}}}
        e^{
            -\frac{1}{2(1 - \gamma)}
            \norm{(y + \mu_t) - \sqrt{\gamma}(y_0 + \mu_0)}_2^2
        } \\
        &=
        \mathcal{N}\big(y+\mu_t; \sqrt{\gamma} (y_0+\mu_0), (1-\gamma) I\big) \\
        &=
        \df_t\big(y + \mu_t | \by_0=y_0 + \mu_0, \bx\big).
    \end{align*}
\end{proof}

Our last results of this form involve the \textit{score} of the density functions $\{\df_t\}$, which play a significant role in diffusion models \cite{song_score-based_2021}. For the forward process $\{\by_t\}$, the score function at time $t$ is given by $\nabla \log \df_t(y | \bx)$, where $\nabla$ represents the gradient with respect to the $y$ variable. This notation means that the gradient of $\log \df_t(\cdot | \bx)$ is evaluated at $y$. For the centered forward process $\big\{ \ty_t \big\}$ the score functions are defined analogously. The following result relates the scores of the two processes.


\begin{lemma}
\label{lemma:score-tilde}
    For all $t\in [0, \ft]$ it holds that $\nabla \log \tf_t(y | \bx) = \nabla \log \df_t(y + \mu_t | \bx)$.
\end{lemma}
\begin{proof}
    By Lemma \ref{lemma:density-tilde} we know $\tf_t(y | \bx) = \df_t(y + \mu_t | \bx)$ for any $y$. This means that:
    \begin{align*}
        \nabla \log \tf_t(y | \bx)
        &=
        \frac{1}{\tf_t(y | \bx)}
        \nabla \tf_t(y | \bx) \\
        &=
        \frac{1}{\df_t(y + \mu_t | \bx)}
        \nabla \tf_t(y | \bx) \\
        &=
        \frac{1}{\df_t(y + \mu_t | \bx)}
        \big(\nabla \df_t(\cdot | \bx) J\big)(y+\mu_t).
    \intertext{Here, $J$ is the Jacobian matrix of the transformation $y\to y+\mu_t$, and hence equal to $I$:}
        &=
        \frac{1}{\df_t(y + \mu_t | \bx)}
        \nabla \df_t(y+\mu_t | \bx) \\
        &=
        \nabla \log \df_t(y + \mu_t | \bx).
    \end{align*}
\end{proof}

Now, consider the forward process $\big\{\ty_t\big\}$ conditioned on the equality $\ty_0 = y_0$. This gives rise to a new set of densities $\tf_t(\cdot | \ty_0=y_0, \bx)$, and thus new score functions $\nabla \log \tf_t(\cdot | \ty_0=y_0, \bx)$. The following result relates the scores of the two forward processes when they are conditioned on the starting variable of the respective process.

\begin{lemma}
\label{lemma:score-tilde-conditioned}
    For all $t\in [0, \ft]$ it holds that
    \begin{equation*}
        \nabla \log \tf_t(y | \ty_0=y_0, \bx)
        = \nabla \log \df_t(y + \mu_t | \by_0=y_0+\mu_0, \bx).
    \end{equation*}
\end{lemma}
\begin{proof}
    We know that both forward processes evolve according to the SDE given by (\ref{eq:forward-sde}). Hence, equation (\ref{eq:forward-density-t-conditioned-on-0}) implies the following equalities, where we use $\gamma$ in place of $\gamma(t, \bx)$ to simplify notation:
    \begin{align*}
        \df_t\big(\cdot | \by_0=y_0 + \mu_0, \bx\big)
        &=
        \mathcal{N}\big(
            \sqrt{\gamma} (y_0 + \mu_0), (1-\gamma) I
        \big), \\
        \tf_t\big(\cdot | \ty_0=y_0, \bx\big)
        &=
        \mathcal{N}\big(
            \sqrt{\gamma} y_0, (1-\gamma) I
        \big).
    \end{align*}
    Then, by replacing the density function of the multivariate Gaussian, we can apply the logarithm and then compute the gradient to get
    \begin{align*}
        \nabla \log \tf_t(&y | \ty_0=y_0, \bx) \\
        &=
        -\frac{1}{(1 - \gamma)} (y - \sqrt{\gamma}y_0) \\
        &=
        -\frac{1}{(1 - \gamma)} (y + \mu_t - \mu_t - \sqrt{\gamma}y_0)
    \intertext{Using Lemma \ref{lemma:mu-time-t} to replace $\mu_t=\sqrt{\gamma}\mu_0$:}
        &=
        -\frac{1}{(1 - \gamma)} (y + \mu_t - \sqrt{\gamma}\mu_0 - \sqrt{\gamma}y_0) \\
        &=
        -\frac{1}{(1 - \gamma)} \big((y + \mu_t) - \sqrt{\gamma}(y_0 + \mu_0)\big) \\
        &=
        \nabla \log \df_t\big(y + \mu_t | \by_0=y_0 + \mu_0, \bx\big).
    \end{align*}
\end{proof}

We can now proceed to the proof of Theorem \ref{thm:KL-upper-bound-inequality}. Recall from Sections \ref{sec:methods-cvdm} and \ref{sec:methods-zmd} that standard CVDM allows us to sample from a distribution $\bz_t \sim \dr_\ft(\cdot | \bx)$. On the other hand, ZMD allows us to sample $\tz_t \sim \tr_\ft(\cdot | \bx)$ and then add $\mu_0$ to $\tz_t$, effectively simulating a different distribution $\tr_\ft^{\mu_0}$. Both $\dr_\ft(\cdot | \bx)$ and $\tr_\ft^{\mu_0}(\cdot | \bx)$ should approximate $\df_0(\cdot | \bx)=\df_\text{data}(\cdot | \bx)$. To study the quality of the approximation, we can analyze the terms $D_\text{KL}(\df_0 || \dr_\ft)$ with $D_\text{KL}(\df_0 || \tr_\ft^{\mu_0})$, where we have omitted the dependence on $\bx$ to simplify notation.

Notice that $\tf_0(\cdot | \bx)$ and $\df_0(\cdot | \bx)$ are related by the following variable transformation: if $\ty_0 \sim \tf_0(\cdot | \bx)$, then $\big(\ty_0 + \mu_0\big) \sim \df_0(\cdot | \bx)$. This is the exact transformation that relates $\tr_\ft$ and $\tr_\ft^{\mu_0}$. Since the Kullback–Leibler divergence $D_\text{KL}$ is invariant under diffeomorphisms \cite{nielsen2019kullback} (such as the $y\to y + \mu_0$ mapping), this means that
\begin{equation}
\label{eq:kl-divergence-invariant}
    D_\text{KL}(\df_0 || \tr_\ft^{\mu_0})
    =
    D_\text{KL}(\tf_0 || \tr_\ft).
\end{equation}

Equation (\ref{eq:kl-divergence-invariant}) allows us to study the quality of approximation for $\df_0(\cdot | \bx)$ by comparing $D_\text{KL}(\df_0 || \dr_\ft)$ with $D_\text{KL}(\tf_0 || \tr_\ft)$. This is simpler than comparing $D_\text{KL}(\df_0 || \dr_\ft)$ with $D_\text{KL}(\df_0 || \tr_\ft^{\mu_0})$, because it involves two similar terms: $D_\text{KL}(\df_0 || \dr_\ft)$ represents how accurately the forward process $\{\by_t\}$ can be reversed, and $D_\text{KL}(\tf_0 || \tr_\ft)$ has the exact same interpretation for the forward process $\big\{\ty_t\big\}$. In a sense, all that matters is how well the respective forward process can be reversed: the operation of subtracting $\mu_0$ and adding it back has no effect in the capacity to approximate $\df_0(\cdot | \bx)$.

We use the following strategy to indirectly compare $D_\text{KL}(\df_0 || \dr_\ft)$ and $D_\text{KL}(\tf_0 || \tr_\ft)$. Although these terms are challenging to analyze, there are known upper bounds for them. In particular, we use the following result \cite{franzese2023much}:
\begin{equation}
\label{eq:KL-upper-bound-preliminary}
    D_\text{KL}(\df_0 || \dr_\ft)
    \leq
    D_\text{KL}(\df_\ft || \mathcal{N}(0, I))
    +
    \frac{1}{2} (I_\nu + I)
\end{equation}
where
\begin{align*}
    I_\nu
    =
    \int_0^\ft \gamma(&t, \bx) \int \bigg( \int \\
    & \big\lVert
        s_\nu(y, t, \bx) - \nabla \log \df_t(y | \by_0 = y_0, \bx)
        \big\rVert_2^2 \\
    & \cdot \df_t(y | \by_0 = y_0, \bx) dy \bigg)
        \df_0(y_0 | \bx) \>dy_0 \>dt, \\
    I
    =
    \int_0^\ft \gamma(&t, \bx) \int \bigg( \int \\
    & \big\lVert
        \nabla \log \df_t(y | \bx) - \nabla \log \df_t(y | \by_0 = y_0, \bx)
        \big\rVert_2^2 \\
    & \cdot \df_t(y | \by_0 = y_0, \bx) dy \bigg)
        \df_0(y_0 | \bx) \>dy_0 \> dt.
\end{align*}
\normalsize
Here, $s_\nu$ represents the learned function that approximates the score, that is, a function that after training should satisfy $s_\nu(y, t, \bx) \approx \nabla \log \df_t(y | \bx)$. There is an explicit relationship between $s_\nu$ and the noise-prediction model $\eps_\nu$ learned by CVDM. Now, we can bound $D_\text{KL}(\df_\ft || \mathcal{N}(0, I))$ by using the derivation in \citet{benton2024a} to change the bound in (\ref{eq:KL-upper-bound-preliminary}) into:
\begin{align}
    D_\text{KL}(\df_0 || \dr_\ft)
    \leq
    &\frac{1}{2} \gamma(\ft, \bx)
    \expec{\df_0(\cdot | \bx)}{\big\lVert\by_0\big\rVert_2^2}
    \nonumber \\
    &+\frac{1}{2} d(1 - \gamma(\ft, \bx)) \label{eq:KL-upper-bound} \\
    &+\frac{1}{2} \big(I_\nu + I\big) \nonumber,
\end{align}
where $d$ denotes the dimension of the $\{\by_t\}$ and $\big\{\ty_t\big\}$ variables. The centered process follows the same dynamics as the standard forward process (that is, equation (\ref{eq:forward-sde})), so analogously, we have
\begin{align}
    D_\text{KL}(\tf_0 || \tr_\ft)
    &\leq
    \frac{1}{2}\gamma(\ft, \bx)
    \expec{\tf_0(\cdot | \bx)}{\big\lVert\ty_0\big\rVert_2^2}
    \nonumber \\
    &+\frac{1}{2} d(1 - \gamma(\ft, \bx)) \label{eq:KL-upper-bound-centered} \\
    &+\frac{1}{2} \big(\tI_\nu + \tI\big) \nonumber.
\end{align}
Here, $\tI_\nu$ and $\tI$ are defined in exactly the same way as $I_\nu$ and $I$, replacing $\df_t$ by $\tf_t$ and using the approximate score function $\ts_\nu(y, t, \bx)$ learned for the centered process $\big\{\ty_t\big\}$. Notice that the right-hand sides of equations (\ref{eq:KL-upper-bound}) and (\ref{eq:KL-upper-bound-centered}) are the upper bounds referred to as $M(\dr_\ft)$ and $M(\tr_\ft^{\mu_0})$ in Section \ref{sec:methods-zmd}, respectively. We now show that these two terms are related such that $M(\tr_\ft^{\mu_0}) \leq M(\dr_\ft)$.

First, suppose an approximate score function $s_\nu$ has been learned for the process $\{\by_t\}$. Then, we can define an approximate score function $\ts_\nu$ for the process $\big\{\ty_t\big\}$ as $\ts_\nu(y, t, \bx) = s_\nu(y + \mu_t, t, \bx)$. Now, the following equality can be derived:
\begin{alignat*}{2}
    \tI_\nu
    &=
    & &\int_0^\ft \gamma(t, \bx) \int \int \bigg( \\
    & 
    & &
    \Big\lVert
        \ts_\nu(y, t, \bx) - \nabla
         \log \tf_t(y | \ty_0 = y_0, \bx)
    \Big\rVert_2^2 \cdot \\
    & 
    & &
    \tf_t(y | \ty_0 = y_0, \bx) dy \bigg)
    \tf_0(y_0 | \bx) dy_0 \> dt
\intertext{Applying Lemmas \ref{lemma:density-tilde}, \ref{lemma:density-tilde-conditioned} and \ref{lemma:score-tilde-conditioned}:}
    &=
    & &\int_0^\ft \gamma(t, \bx) \int \int \bigg( \\
    & 
    & &
    \Big\lVert
        \ts_\nu(y, t, \bx) - \nabla
        \log \df_t(y + \mu_t| \by_0 = y_0 + \mu_0, \bx)
    \Big\rVert_2^2 \cdot \\
    & 
    & &
    \df_t(y + \mu_t | \by_0 = y_0 + \mu_0, \bx) dy \bigg)
    \df_0(y_0 + \mu_0 | \bx) dy_0 \> dt
\intertext{Replacing $\ts_\nu(y, t, \bx) = s_\nu(y + \mu_t, t, \bx)$ and using the substitutions $u = y+\mu_t$ and $u_0 = y_0+\mu_0$ for the inner integrals:}
    &=
    & &\int_0^\ft \gamma(t, \bx) \int \int \bigg( \\
    & 
    & &
    \Big\lVert
        s_\nu(u, t, \bx) - \nabla
        \log \df_t(u | \by_0 = u_0, \bx)
    \Big\rVert_2^2 \cdot \\
    & 
    & &
    \df_t(u | \by_0 = u_0, \bx) dy \bigg)
    \df_0(u_0 | \bx) du_0 \> dt \\
    & = & & I_\nu. 
\end{alignat*}

Analogously we can prove that $\tI = I$, where the only difference in the proof is that we need to use Lemma \ref{lemma:score-tilde} to replace $\nabla \log \tf_t(y | \bx) = \nabla \log \df_t(y + \mu_t | \bx)$ before performing the substitutions inside the inner integrals. Notice what this implies. On the one hand, $\tI = I$ always holds. On the other hand, for any learned score $s_\nu$, we can construct $\ts_\nu$ in such a way that $\tI_\nu = I_\nu$. This means that in the upper bounds $M(\dr_\ft)$ and $M(\tr_\ft^{\mu_0})$, there is no general difference that can be established between the terms $(I_\nu + I)$ and $(\tI_\nu + \tI)$. Furthermore, the $d(1 - \gamma(\ft, \bx))$ term is exactly the same in both cases too. Hence, to distinguish between both upper bounds, we need to focus on the terms:
\begin{equation*}
    \expec{\df_0(\cdot | \bx)}{\big\lVert\by_0\big\rVert_2^2},
    \>\>
    \expec{\tf_0(\cdot | \bx)}{\big\lVert\ty_0\big\rVert_2^2},
\end{equation*}
which can be compared via the following derivation
\begin{align*}
    \mathbb{E}_{\tf_0(\cdot | \bx)} &\left[
        \big\lVert\ty_0\big\rVert_2^2
    \right]    \\
    &=
    \int \norm{y_0}_2^2 \> \tf_0(y_0 | \bx) \> dy_0
\intertext{By Lemma \ref{lemma:density-tilde}:}
    &=
    \int \norm{y_0}_2^2 \> \df_0(y_0 + \mu_0 | \bx) \> dy_0
\intertext{Using the substitution $u_0 = y_0 + \mu_0$ for the integral:}
    &=
    \int \norm{u_0 - \mu_0}_2^2 \> \df_0(u_0 | \bx) \> du_0 \\
    &=
    \int \left(\norm{u_0}_2^2 + \norm{\mu_0}_2^2 - 2\mu_0^T u_0 \right) \> \df_0(u_0 | \bx) \> du_0 \\
    &=
    \int \norm{u_0}_2^2 \df_0(u_0 | \bx) du_0
    + \norm{\mu_0}_2^2 \int \df_0(u_0 | \bx) du_0 \\
    &\hspace{97pt}- 2 \mu_0^T \int u_0 \> \df_0(u_0 | \bx) du_0 \\
    &=
    \expec{\df_0(\cdot | \bx)}{\big\lVert\by_0\big\rVert_2^2}
    + \norm{\mu_0}_2^2 - 2 \norm{\mu_0}_2^2 \\
    &=
    \expec{\df_0(\cdot | \bx)}{\big\lVert\by_0\big\rVert_2^2}
    - \norm{\mu_0}_2^2 \\
    &\leq
    \expec{\df_0(\cdot | \bx)}{\big\lVert\by_0\big\rVert_2^2}.
\end{align*}
This finishes the proof of Theorem \ref{thm:KL-upper-bound-inequality}. When comparing the bounds in equations (\ref{eq:KL-upper-bound}) and (\ref{eq:KL-upper-bound-centered}), we have shown that all terms are or can be made equal, except for the second moment terms, for which the above inequality shows that the standard (non-centered) process has a value at least as big as the centered process. In other words, we conclude that $M(\tr_\ft^{\mu_0}) \leq M(\dr_\ft)$.

\begin{figure*}
    \centering
    \includegraphics[width=0.9\textwidth]{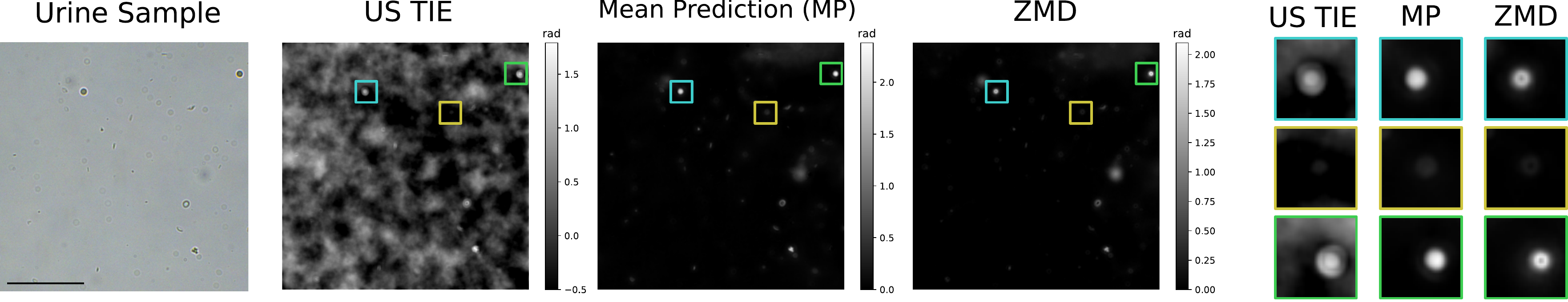}
    \caption{Urine sample with red blood cells. From left to right: US TIE phase reconstruction, Mean Prediction, and Zero-Mean Diffusion prediction. Regions of the image are enlarged to better show the reconstruction details. Scale bar is 200 $\mu$m.}
    \label{fig:rbcs}
\end{figure*}

\section{Additional Experiments}
\label{sec:additional-experiments}

To further validate our method, we conducted additional experiments by acquiring three through-focus stacks of red blood cells using a microscope equipped with a plan objective. We perform comparative analyses between our method and the US TIE technique. The results demonstrate that our method successfully recovers structural details that are not captured by any of the US TIE reconstructions. For instance, notice the black dots at the center of red blood cells in the enlarged images from Figure \ref{fig:rbcs}, successfully recovered by ZMD but absent in the US TIE solution. Our method thus provides enhanced accuracy and reliability in the reconstruction of red blood cell morphology. These findings underscore the superior performance of ZMD in extracting critical structural information. Figure \ref{fig:qpi-rbc} shows the results of the different methods, and Tables \ref{fig:table-ssim-rbc} and \ref{fig:table-mae-rbc} present the MS-SSIM and MAE scores for each sample and method, with sample numbers corresponding to the rows in Figure \ref{fig:qpi-rbc}.

\begin{figure}
    \centering
    \includegraphics[width=0.9\columnwidth]{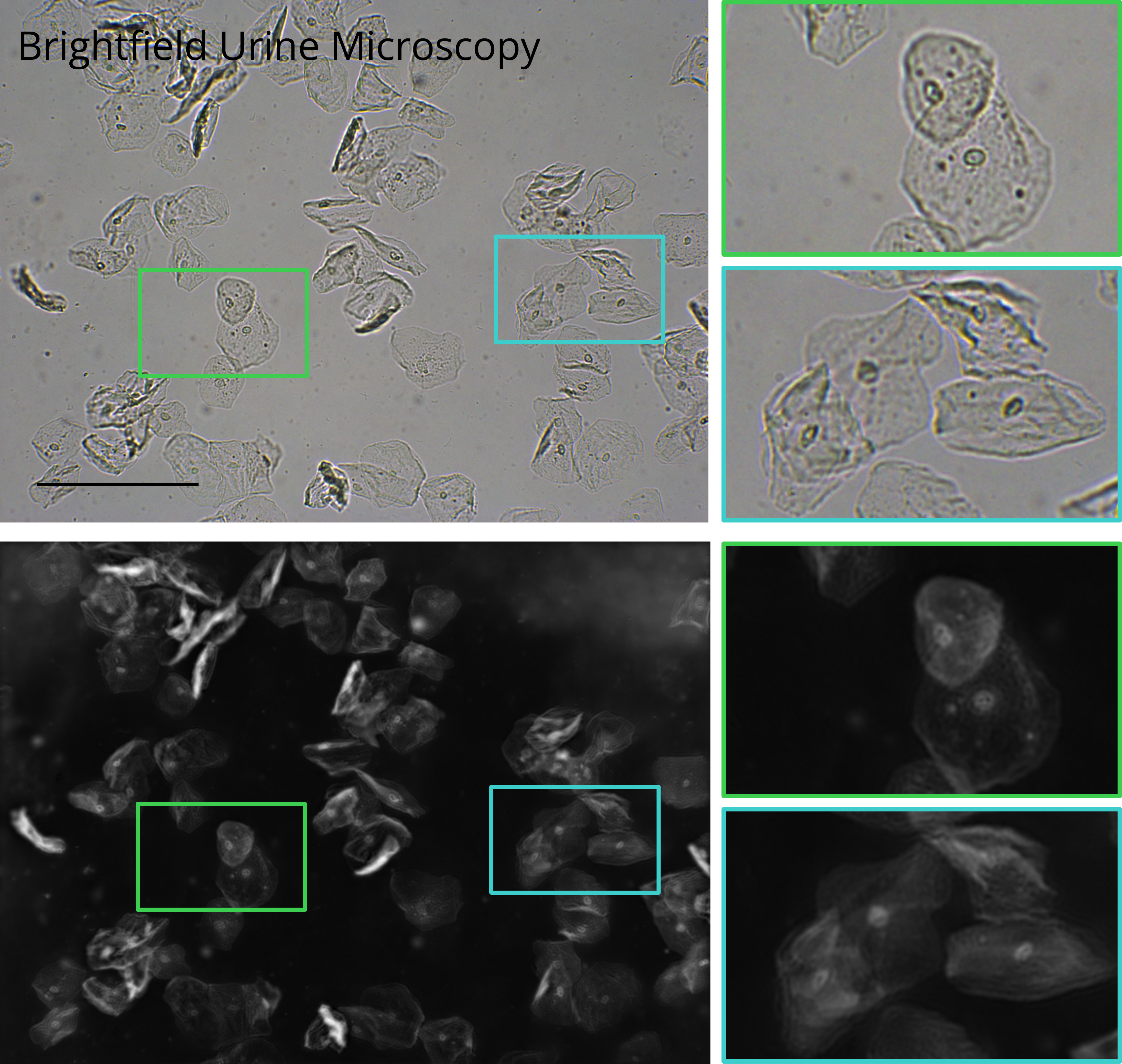}
    \caption{Phase quantification performed on clinical images. Regions of the images are enlarged to better show the objects in the image. Scale bar is 200 $\mu$m.}
    \label{fig:uti-2}
\end{figure}

\begin{figure}
  \centering
  \includegraphics[width=0.9\columnwidth]{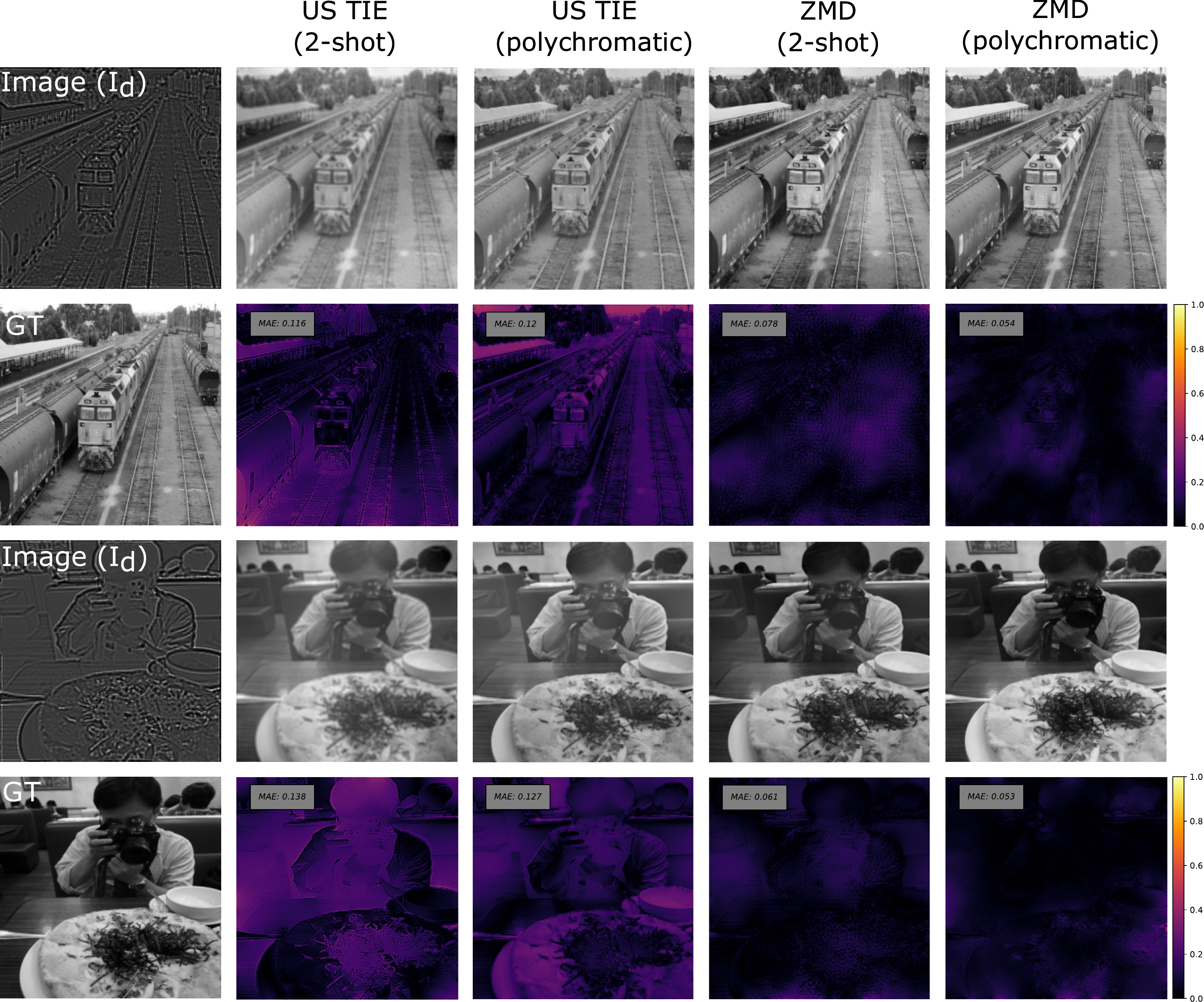}
  \caption{QPI methods assessed using synthetic images from HCOCO.}
  \label{subfig:qpi-syn-2}
\end{figure}

\begin{figure}
  \centering
  \includegraphics[width=0.9\columnwidth]{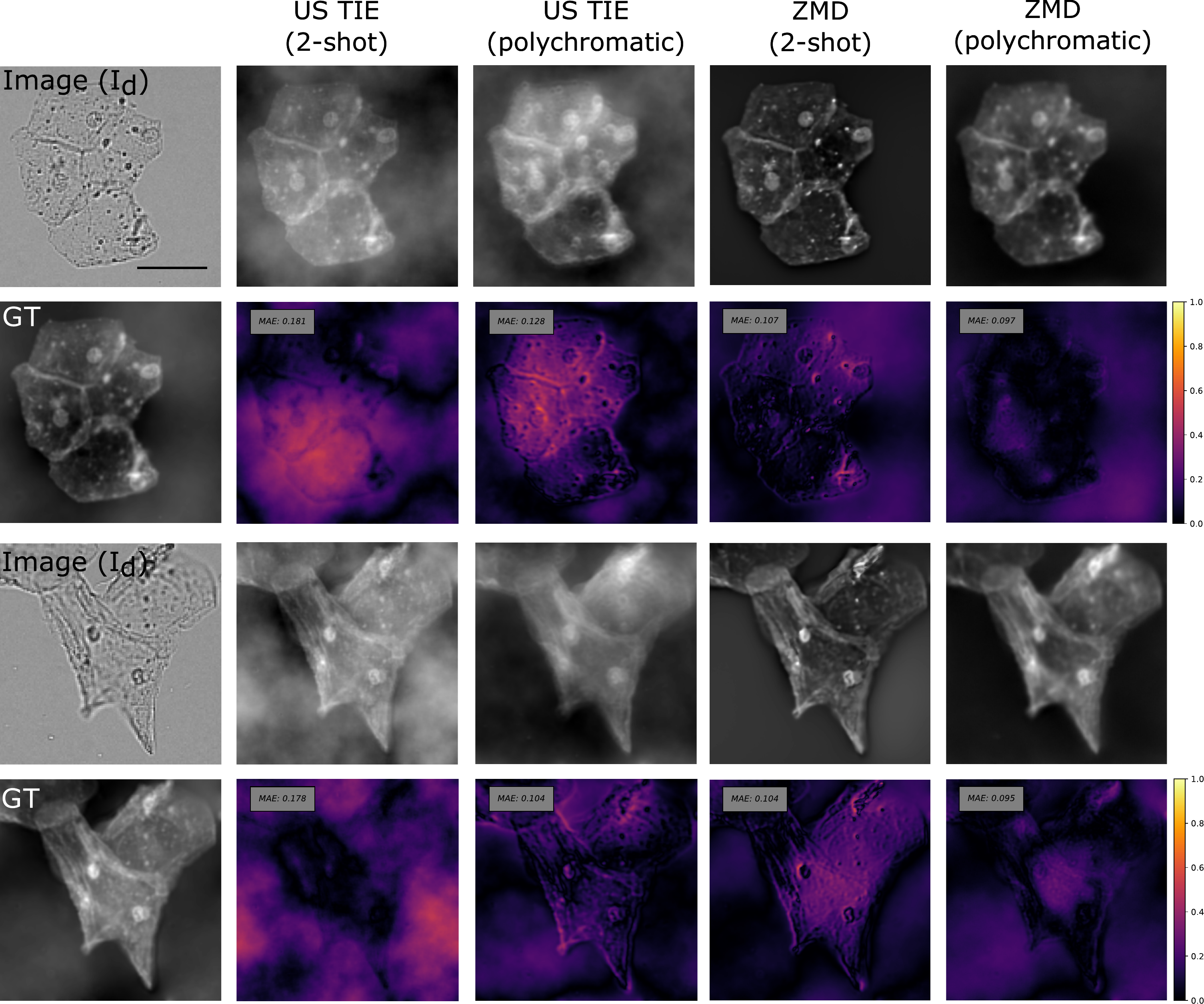}
  \caption{QPI methods assessed using clinical images. Scale bar (upper leftmost image) is 50 $\mu$m.}
  \label{subfig:qpi-real-2}
\end{figure}

\begin{table*}
\caption{Performance metrics for QPI Through-Focus Brightfield Images for red blood cells. ``2s'' stands for the 2-shot modality (using two images to measure the phase), while ``p'' stands for the polychromatic (single exposure) modality.}
    \centering
    \begin{subtable}{0.47\textwidth}
        \scriptsize
        \centering
        \caption{MS-SSIM on Through-Focus Brightfield Images.}
        \label{fig:table-ssim-rbc}
        \begin{tabular}{llllll}
            \toprule
            Sample & US TIE (2s) & US TIE (p)& MP (p) & ZMD (2s) & ZMD (p) \\
            \midrule
            1 & \textbf{0.90} & 0.67 & 0.77 & 0.83 & 0.80 \\
            2 &  0.77 &   0.52&0.78 & 0.71 &\textbf{0.83} \\
            3  & 0.79 & 0.47 & 0.88& 0.66 & \textbf{0.90} \\
            \bottomrule
        \end{tabular}
    \end{subtable}%
    \hfill
    \begin{subtable}{0.47\textwidth}
        \scriptsize
        \centering
        \caption{MAE on Through-Focus Brightfield Images.}
        \label{fig:table-mae-rbc}
        \begin{tabular}{llllll}
            \toprule
            Sample & US TIE (2s) & US TIE (p) &MP (p) & ZMD (2s) & ZMD (p) \\
            \midrule
            1 & 0.11 & 0.17&0.09 & 0.08 & \textbf{0.08} \\
            2 & 0.08 & 0.24&\textbf{0.06}  & 0.08 & \textbf{0.06} \\
            3 & 0.06 & 0.25& \textbf{0.06} & 0.14 & \textbf{0.06} \\
            \bottomrule
        \end{tabular}
    \end{subtable}
    
    \label{fig:combined-tables-rbc}
\end{table*}

\begin{figure*}
    \centering
    \includegraphics[width=0.75\textwidth]{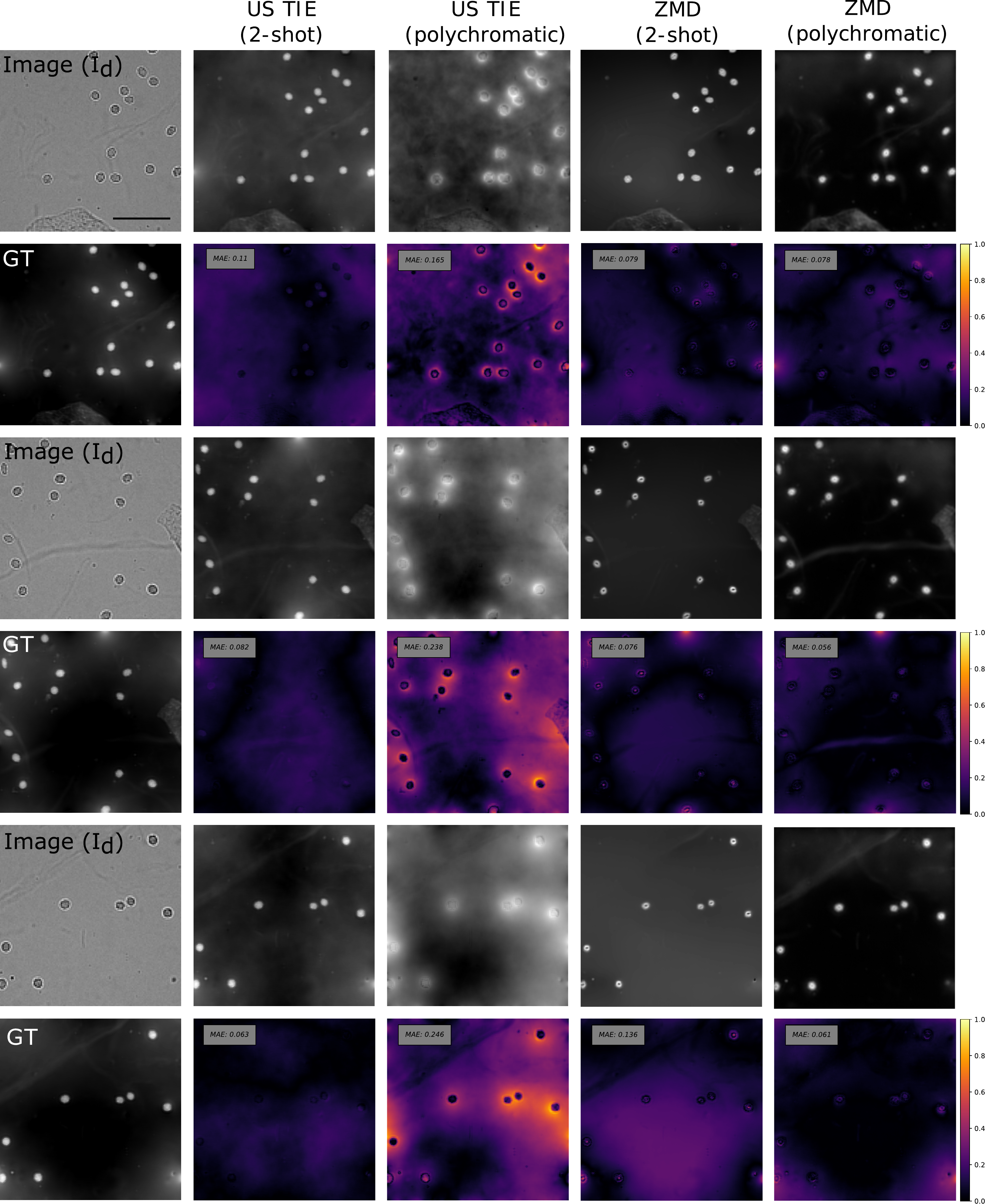}
    \caption{Comparison of QPI methods. Each sample is composed from left to right: image with defocus $d$ ($I_d$) with corresponding ground truth (GT) below; reconstructions and error images follow. All phase images are quantitatively interpretable, with values between 0 and 2.5 radians. Scale bar is 50 $\mu$m.}
    \label{fig:qpi-rbc}
\end{figure*}

\newpage
\onecolumn

\section{Implementation Details}

\subsection{Hardware Used}
All the experiments were conducted on a system running Ubuntu 18.04 with an Intel i7-12700KF processor (12 cores), 32GB of main memory, 1TB of swap memory, and one NVIDIA A6000 graphics card. We implemented our model using TensorFlow 2.10 \cite{tensorflow2015-whitepaper}.


\subsection{Training and Inference Algorithms}
Algorithm \ref{alg:training_pdm1} describes the training of the proposed algorithm, while Algorithm \ref{alg:training_pdm2} demonstrates the inference process.
For the noise-prediction model we employ the same architecture used in SR3 \cite{saharia_palette_2022}, and for the mean model $\mu_\theta$ we use a traditional U-Net \cite{ronneberger2015unetconvolutionalnetworksbiomedical}.

A key hyper-parameter for diffusion models is the number of timesteps $T$. The values reported in the literature are generally in the $[100, 200]$ range \cite{saharia_palette_2022, dhariwal_diffusion_2021, ho_denoising_2020}. Based on our experience, this allows for good reconstruction quality while keeping inference times at reasonable levels \cite{dellamaggiora2024conditional}. We tested values of $T=100$, $T=150$, and $T=200$. We found that the best results in terms of average error metrics were obtained with $T=200$, which is the configuration we use for all metrics reported in this work. Empirical variance of error metrics was comparable for all three values. On another front, the loss function for ZMD depends on a hyper-parameter $\omega$ (see equation \eqref{eqn:loss-zmd}). We compared the results obtained with $\omega=1$ and $\omega=2$, keeping $\omega=2$ as it provided marginally better results.

\begin{figure*}[ht]
\centering
\SetAlFnt{\small}
\SetAlCapFnt{\small}
\SetAlCapNameFnt{\small}

\begin{algorithm}[H]
\caption{Training of Noise-Prediction Model $\tilde{\eps}_\nu(\tz_t(\epsilon), t,\bx)$}
\label{alg:training_pdm1}
\SetKwInOut{Input}{Input}
\SetKwInOut{Output}{Output}

\Repeat{\textnormal{converged}}{
    $(\bx, \by) \sim \df_\text{data}(\bx,\by)$ \\
    $t \sim \text{Unif}([0, 1])$ \\
    $\epsilon \sim \mathcal{N}(0, I)$ \\
    $\ty \leftarrow \by - \text{stop gradient}(\mu_\theta(X))$ \\
    \text{Take a gradient descent step} \\
    
    \Indp $\nabla_{\xi} \gL_\text{ZMD}(\gamma, \bx, \ty ,\by, \epsilon)$ \text{where $\xi$ are the parameters of the joint model. See equation (\ref{eqn:loss-zmd}) and related terms.} \\
    \Indm
}
\end{algorithm}

\vspace{0.5em}

\begin{algorithm}[H]
\caption{Inference in $T$ timesteps for input $\bx$, schedule function $\beta$ and learned models $\mu_\theta$, $\tilde{\eps}_\nu$}
\label{alg:training_pdm2}
\SetKwInOut{Input}{Input}
\SetKwInOut{Output}{Output}

\For{$t = 0$ to $T$}{
    $\betaF_t \leftarrow \frac{\betaFtx{t/T}{\bx}}{T}$\\
    $\alpha_t \leftarrow 1-\betaF_t$\\
    $\gamma_t \leftarrow \prod_{s=1}^t \alpha_{s}$\\
}
\vspace{2pt}
$\ty_T \sim \mathcal{N}(0, I)$ \\
\vspace{2pt}
\For{$t = T$ to $1$}{
    \lIf{$t = 1$}{$\epsilon \leftarrow 0$}
    \lElse{$\epsilon \sim \mathcal{N}(0, I)$}

    $\ty_{t-1} \leftarrow \frac{1}{\sqrt{\alpha_t}}\bigg(\ty_t - \frac{\betaF_t}{\sqrt{\mathbf{1}-\gamma_t}}
    \tilde{\eps}_\nu(\ty_t,t,\bx)\bigg)+ \sqrt{\betaF_t}\epsilon$
}

\Return $\ty_0 + \mu_\theta(\bx)$

\end{algorithm}

\caption{Training and Sampling algorithms for ZMD.}
\label{alg:train-inference}
\end{figure*}
\newpage
\subsection{Dataset Simulation Pipeline}
This section provides a graphical explanation of our simulation pipeline, as described in Section 3. The diagram illustrates the key steps involved in simulating polychromatic acquisitions, including processing ImageNet images, phase object modeling, defocus simulation, wavelength-specific propagation, and incorporating sensor characteristics. This visualization complements the detailed mathematical description and highlights the methodology used to construct the dataset for training phase reconstruction models.

\begin{figure}[ht]
    \centering
    \includegraphics[width=\linewidth]{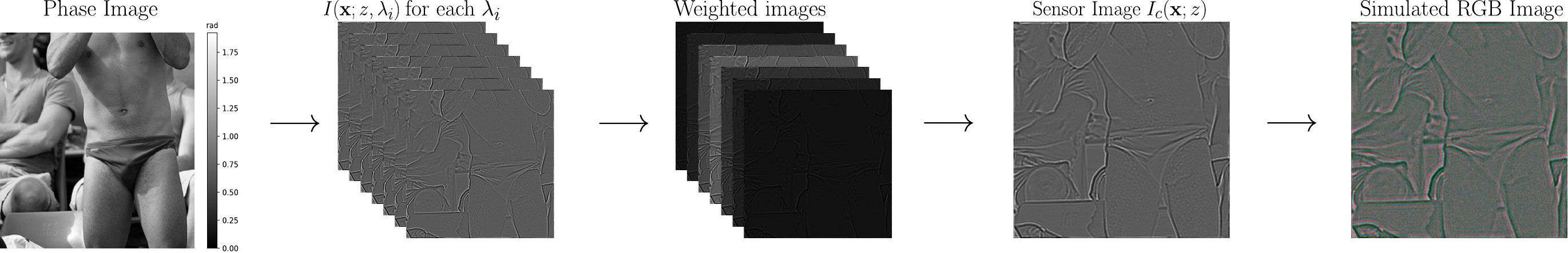}
    \caption{Overview of the simulation pipeline. The process begins with a randomly selected ImageNet image, converted to grayscale and linearly mapped to a phase range of $0$ to $3.5$ radians (\textit{Phase Image}). Using equation (2), the resulting image $I(\vx; z, \lambda_i)$ is simulated under a fixed defocus distance $z$, randomly selected between $0.1\mu$m and $3\mu$m using the Fresnel propagator for wavelengths ranging from $400$nm to $700$nm with equal contribution, discretized at $6$nm increments. Subsequently, the Power Spectral Density (PSD), $S_\lambda(\mathbf{u}) = \delta_\lambda(\mathbf{u}) Q_c(\lambda)$, is applied to generate the \textit{Weighted Images}, the parameter $\sigma_c$ is uniformly sampled between $0.01$ and $0.1$. The discretized integral is solved to compute the sensor image $I_c(\vx; z)$. Finally, this process is repeated for each RGB channel to synthesize the \textit{Simulated RGB Image}.}
    \label{fig:sim-pipeline}
\end{figure}

\newpage
\onecolumn

\section{Reproducibility Checklist}
\label{sec:rep-checklist}

This paper
\begin{itemize}
    \item Includes a conceptual outline and/or pseudocode description of AI methods introduced: \textbf{yes}
    \item  Clearly delineates statements that are opinions, hypothesis, and speculation from objective facts and results: \textbf{yes}
    \item Provides well marked pedagogical references for less-familiar readers to gain background necessary to replicate the paper: \textbf{yes}
\end{itemize}

\noindent
Does this paper make theoretical contributions? \textbf{yes}

\noindent
If yes, please complete the list below.
\begin{itemize}
    \item All assumptions and restrictions are stated clearly and formally: \textbf{yes}
    \item All novel claims are stated formally (e.g., in theorem statements): \textbf{yes}
    \item Proofs of all novel claims are included: \textbf{yes}
    \item Proof sketches or intuitions are given for complex and/or novel results: \textbf{yes}
    \item Appropriate citations to theoretical tools used are given: \textbf{yes}
    \item All theoretical claims are demonstrated empirically to hold: \textbf{partial. We did not test the exact theoretical propositions, but in our view the empirical observations (i.e., ZMD converging to useful results, while CVDM needs normalized data to work) are in line with, and provide support for, our theoretical claims.}
    \item All experimental code used to eliminate or disprove claims is included: \textbf{NA}
\end{itemize}

\noindent
Does this paper rely on one or more datasets? \textbf{yes}

\noindent
If yes, please complete the list below.
\begin{itemize}    
    \item A motivation is given for why the experiments are conducted on the selected datasets: \textbf{yes}
    \item All novel datasets introduced in this paper are included in a data appendix: \textbf{no}
    \item All novel datasets introduced in this paper will be made publicly available upon publication of the paper with a license that allows free usage for research purposes: \textbf{yes}
    \item All datasets drawn from the existing literature (potentially including authors’ own previously published work) are accompanied by appropriate citations: \textbf{yes}
    \item All datasets drawn from the existing literature (potentially including authors’ own previously published work) are publicly available: \textbf{yes}
    \item All datasets that are not publicly available are described in detail, with explanation why publicly available alternatives are not scientifically satisficing: \textbf{yes}
\end{itemize}

\noindent
Does this paper include computational experiments? \textbf{yes}

\noindent
If yes, please complete the list below.
\begin{itemize}
    \item Any code required for pre-processing data is included in the appendix: \textbf{no}.
    \item All source code required for conducting and analyzing the experiments is included in a code appendix: \textbf{no}.
    \item All source code required for conducting and analyzing the experiments will be made publicly available upon publication of the paper with a license that allows free usage for research purposes: \textbf{yes}
    \item All source code implementing new methods have comments detailing the implementation, with references to the paper where each step comes from: \textbf{yes}
    \item If an algorithm depends on randomness, then the method used for setting seeds is described in a way sufficient to allow replication of results: \textbf{yes}
    \item This paper specifies the computing infrastructure used for running experiments (hardware and software), including GPU/CPU models; amount of memory; operating system; names and versions of relevant software libraries and frameworks: \textbf{yes}
    \item This paper formally describes evaluation metrics used and explains the motivation for choosing these metrics: \textbf{yes}
    \item This paper states the number of algorithm runs used to compute each reported result: \textbf{yes}
    \item Analysis of experiments goes beyond single-dimensional summaries of performance (e.g., average; median) to include measures of variation, confidence, or other distributional information: \textbf{partial, as standard deviation was included for synthetic dataset (see Table 1). Also, to further validate generalization power we tested our model in samples acquired specifically for this paper, and we tested our model for two different types of microscope with uncurated images.}
    \item The significance of any improvement or decrease in performance is judged using appropriate statistical tests (e.g., Wilcoxon signed-rank): \textbf{partial, hypothesis testing was performed for the synthetic data. For the real microscope samples it is not possible to perform testing as the set corresponds to a few samples specifically acquired for this problem.}
    \item This paper lists all final (hyper-)parameters used for each model/algorithm in the paper’s experiments: \textbf{yes}
    \item This paper states the number and range of values tried per (hyper-) parameter during development of the paper, along with the criterion used for selecting the final parameter setting: \textbf{yes}
\end{itemize}

\end{document}